\documentclass[11pt,letterpaper]{article}
\usepackage{amsmath,amsthm,amssymb,amsfonts,color}
\usepackage{xcolor}
\usepackage{authblk}
\usepackage{natbib}
\usepackage[margin=1in]{geometry}
\usepackage{booktabs} % For formal tables
\usepackage[ruled]{algorithm2e} % For algorithms

\usepackage{algorithmic}
\usepackage{enumitem}
\SetAlFnt{\small}
\SetAlCapFnt{\small}
\SetAlCapNameFnt{\small}
\SetAlCapHSkip{0pt}
\IncMargin{-\parindent}

\theoremstyle{plain}
\newtheorem{theorem}{Theorem}[section]
\newtheorem{lemma}[theorem]{Lemma}

\newtheorem{proposition}[theorem]{Proposition}

\theoremstyle{definition}
\newtheorem{definition}[theorem]{Definition}
\newtheorem{assumption}[theorem]{Assumption}

\theoremstyle{remark}

\newtheorem*{remark*}{Remark}

\newcommand{\exan}{ex-ante\xspace}

\newcommand{\vq}{v_{\tt q}}

\newcommand{\selection}{{\sc ActionThresholdSelection}}
\newcommand{\spa}{{\tt SPA}}
\newcommand{\fpa}{{\tt FPA}}

\DeclareMathOperator*{\E}{\mathbb{E}}
\DeclareMathOperator*{\calU}{\mathcal{U}}
\DeclareMathOperator*{\argmax}{argmax}

\begin{document}
\title{Auctions with Contract Design}

\author[1]{Xiaolin Bu}
\author[1]{Jiarong Jin}
\author[1]{Junzhu Ke}
\author[2]{Pinyan Lu}
\author[1]{Biaoshuai Tao}
\author[3]{Xiang Yan}
\author[3]{Chunxue Yang}
\author[3]{Haikuo Yang}
\author[4]{Zhihua Zhu}

\affil[1]{Shanghai Jiao Tong University: \texttt{lin\_bu@sjtu.edu.cn},
\texttt{jinjiarong@sjtu.edu.cn},
\texttt{jocelyn-10909@sjtu.edu.cn},
\texttt{bstao@sjtu.edu.cn}}
\affil[2]{Shanghai University of Finance and Economics: \texttt{lu.pinyan@mail.shufe.edu.cn}}
\affil[3]{Huawei Taylor Lab:
\texttt{xyansjtu@163.com},
\texttt{yangchunxue7@huawei.com},
\texttt{yanghaikuo1@huawei.com}}
\affil[4]{Huawei Ads: \texttt{zachzhu0418@126.com}}

\date{}

\maketitle

\begin{abstract}

We consider a new auction model where the bidders' valuations and the auctioneer's revenue/utility depend on a \emph{quality} factor of the transaction determined by costly and strategic investments of the bidders. A canonical application of our model is ad auctions, where the high quality of the ads benefits both the advertisers and the platform's long-term revenue due to good user experiences. This model also captures many other practical scenarios, such as government concessions and crowdsourcing contests. Crucially, these quality-enhancing efforts made by the bidders are often \emph{sunk costs} incurred prior to the allocation, creating a fundamental moral hazard problem where the risk of losing the auction discourages investments. In this paper, we study the design of revenue-maximizing contracts integrated into auctions: the auctioneer commits to a transfer rule that rewards the winner for the ex-post realized quality of the transaction to incentivize higher effort. Our new framework is a natural generalization of both the auction theory and the principal-agent model.

We consider both the second-price and the first-price auctions, and we consider the standard Bayesian setting with private bidder types. We characterize Bayes Nash equilibria and show that natural symmetric equilibria exist in both auctions. Assuming these natural equilibria are played by the bidders and the number of bidders is large, we study linear contracts and derive the optimal reward factor of the transfer rule that maximizes the auctioneer's revenue. As the main result, we show that the optimal reward factor converges to the auctioneer's marginal benefit from the quality, as the number of bidders grows. That is, \emph{it is optimal for the auctioneer to ``fully pass through'' the quality value to the winner}. This observation is largely independent of the auction rule used: we derive a revenue equivalence theorem showing that the revenue remains the same as long as symmetric Bayes Nash equilibria exist. Lastly, by quantitatively comparing with the standard auctions where no quality reward is used, we show that the use of contracts effectively improves the revenue by incentivizing high investments from the bidders.
\end{abstract}

\section{Introduction}
Consider the following auction setting. 
An auctioneer seeks to allocate a task or opportunity to one of several capable agents.
An auction is used to decide which bidder will execute the task. 
Crucially, each bidder makes a strategic investment before participating in the auction. 
The quality of the task completion depends on the amount of investment sunk by the winner, and the realized quality affects the utilities of both the winner and the auctioneer. 
Consequently, the goal of the auctioneer is to maximize their total payoff/revenue, which depends on both the monetary payment determined by the auction and the realized quality of the task completion.

One typical application of this model is the ad auction~\citep{edelman2007internet,varian2007position}.
The platform, which is the auctioneer, would like to sell one ad slot to one of the advertisers, which are the bidders.
The advertisers need to exert costly efforts to design high-quality creatives or optimize landing pages. 
This investment is beneficial for both the advertisers and the platform.
Naturally, high-quality ads increase the click-through and conversion rates for the advertiser.
On the other hand, high-quality ads benefit the platform by enhancing the user experience and long-term ecosystem health, and incentivizing such ads to bolster long-term platform utility has garnered significant attention from both academia and industry~\citep{athey2011position,abrams2007ad,li2023optimally,yan2025optimal}.

This model also has a wide range of applications in other fields, such as government concessions~\citep{hendricks1988empirical,cramton1997fcc} and crowdsourcing/innovation contests~\citep{dipalantino2009crowdsourcing,moldovanu2001optimal,archak2009optimal}.
In the allocation of government concessions, such as oil and gas leases or spectrum licenses, the government acts as the auctioneer and sells rights to the highest bidder. 
Similar to the ad auction setting, bidders must incur significant amounts of investment, such as geological surveys or network feasibility studies, prior to the auction. 
The government seeks to maximize the earnings from the bidders while ensuring high-quality execution (e.g., environmental safety or network reliability), which depends on these upfront investments. 
In crowdsourcing/innovation contests hosted by platforms such as Kaggle and InnoCentive, participants need to invest significant amounts of effort to design algorithms/solutions before a winner is selected.
The platform seeks good solutions and designs prize schemes that reward participants based on the qualities of the solutions.\footnote{There is one small difference between the crowdsourcing contests setting and the model studied in our paper. The auctions used in these contests are normally reverse auctions where the auctioneer pays the bidders. On the other hand, our paper deals with forward auctions where the bidders pay the auctioneer. However, these two auction settings are structurally identical.}

In all the above-mentioned scenarios, it is important to notice that the investments made by the bidders are prior to the auctions, and are therefore \emph{sunk cost}.
Bidders must take the risk of losing the auction and making the investments in vain.
Advertisers have already spent effort on ad designs before entering the auction, bidders in government concessions perform geological surveys or network feasibility studies prior to the auction, and participants of a crowdsourcing contest need to work out the algorithms/solutions before the submissions.
This distinguishes our work from previous work by~\citet{che1993design,branco1997design,cripps1994design,mcafee1986bidding,asker2010procurement}, which models costly actions in auctions.

In addition, the investments made by the bidders are typically hidden to the auctioneer (instead, the auctioneer normally only sees the quality of the outcome).
This introduces a fundamental \emph{moral hazard} issue where the risk of losing the auction discourages investments and the actions of the investments are hidden.
As \citet{holmstrom1979moral} remarked, \emph{it has long been recognized that a problem of moral hazard may arise when individuals engage in risk sharing under conditions such that their privately taken actions affect the probability distribution of the outcome.}
This moral hazard issue is addressed in the \emph{contract theory} literature, which has been extensively studied by economists and computer scientists (see a recent survey by~\citet{dutting2024algorithmic}).
In the principal-agent model, the principal uses \emph{contracts}, which design payment rules based on the observable quality of the outcomes to reward the agent, to incentivize the agent's effort investment.

In this paper, we address this challenge by integrating a quality-contingent contract into an auction.
In our mechanism, the winner pays according to the auction rule (e.g., the second-highest bid in the second-price auction) but also receives a reward from the auctioneer based on the ex-post observed quality.
By integrating a contract into the auction, the auctioneer's \emph{revenue} consists of three parts: the payment from the winner, plus the long-term benefit score based on the quality of the outcome, and minus the quality reward to the winner according to the contract.
Compared with the classical auctions without contracts, the extra third component of the revenue incurs a loss in the auctioneer's revenue, and this can be potentially compensated by better-quality outcomes due to the higher investments (incentivized by the contract) of the bidders.
Notice that a better-quality outcome improves both the first and the second components of the revenue:
bidders place higher bids by knowing that they will be rewarded should they win, and the long-term benefit score depends directly on the quality of the outcome.
We answer the following three fundamental questions in this paper:
\begin{enumerate}
    \item How will the bidders react and play when a contract is added to an action?
    \item How to design optimal contracts that maximize the auctioneer's revenue?
    \item Compared with the standard auction where no contract is used, can the auctioneer improve the overall revenue by sacrificing a little bit on rewarding the winner? If so, by how much?
\end{enumerate}

From the perspective of contract theory, our framework can be viewed as a natural generalization of the classical principal-agent model to a \emph{competitive environment}. 
Traditional contract theory~\citep{holmstrom1979moral,grossman1992analysis} typically focuses on a bilateral relationship where the principal interacts with a single pre-selected agent.
In that setting, the challenge in mechanism design is restricted to overcoming the moral hazard issue subject to the agent's fixed participation constraint. 
Our work extends this to a multilateral setting where the principal must simultaneously solve a selection problem (identifying the agent with the highest capability) and an incentive problem (motivating hidden effort).

\subsection{Our Results}
We consider the standard Bayesian setting where each bidder has a private type that is independently sampled from a public distribution, and we use the solution concept of the Bayes Nash equilibrium (in its interim form).
We address all three questions above.

To address the first question, for the second-price auction, we show that Bayes Nash equilibria always exist, and each bidder always plays a \emph{threshold strategy} in any equilibrium, where a threshold strategy consists of a threshold for the type such that a bidder plays the more costly action yielding better expected quality if and only if her type is above this threshold.
When bidders' type distributions are identical, we show that there exists a unique symmetric Bayes Nash equilibrium where bidders use the same threshold in their threshold strategy.
We also show that a similar type of symmetric Bayes Nash equilibria exists for the first-price auction.

For the corresponding revenue maximization problem, we consider linear contracts and show that the optimal reward factor converges to the auctioneer's marginal benefit from the quality as the number of bidders grows.
This implies that \emph{it is optimal for the auctioneer to ``fully pass through'' her quality value to the winner}.
This is explicitly derived from the second-price auction.
We also prove a revenue equivalence theorem, which states that the revenue remains the same regardless of the auction rule, as long as a symmetric Bayes Nash equilibrium exists.
Therefore, the optimal contract we derive applies to a broad class of auctions where symmetric equilibria exist.
In particular, since we have also shown that a symmetric Bayes Nash equilibrium exists in the first-price auction, the optimal contract applies to the first-price auction.

Finally, we demonstrate that the derived optimal contracts always result in better revenue compared with the auction where no contract is used.
In addition, compared to standard auctions, the use of optimal contracts yields strict revenue improvements that are almost proportional to the potential gain in the quality score, except in one corner case where the auctioneer does not have much incentive for improving the quality (Theorem~\ref{thm:BNE-optimalt} and the discussions following it give a precise characterization of this corner case).
Specifically, it is proportional up to a logarithmic additive term.
This answers the third question affirmatively: \emph{the use of the contract can always promote high qualities, and the resultant revenue gain can compensate for the loss in rewarding the winner.}

\subsection{Related Work}
\paragraph{Multidimensional auctions and procurements.}
Our model shares structural similarities with the literature on multidimensional auctions/procurement, where bidders compete on both price and non-monetary attributes (quality). 
The seminal work by \citet{che1993design} analyzes a procurement setting where bidders bid on both price and quality, and the winner is determined by a \emph{scoring rule}. 
Che characterizes the optimal scoring rule, showing that the auctioneer should bias the score to limit the information rents of efficient suppliers. 
This framework has been extended to settings with correlated costs \citep{branco1997design} and various procurement contexts \citep{asker2010procurement,cripps1994design}.
Similar to our paper, \citet{mcafee1986bidding} also integrate contracts into the government procurement setting and characterize optimal contracts.

However, a critical distinction separates our work from these papers is that quality is typically modeled as a verifiable \emph{commitment} made at the bidding stage and efforts are invested in ex-post (i.e., as a part of the utility function only for the winner). 
In contrast, we model quality as the result of a sunk investment made \emph{prior} to the auction. 
This timing introduces a fundamental ``hold-up" problem: bidders must invest under the uncertainty of winning, creating a moral hazard friction that those above-mentioned papers do not address. 
In this sense, our work aligns more closely with the literature on auctions with pre-bid investments, as we will describe next.
%such as \citet{tan1992entry} and \citet{piccione1996cost}, who study how auction formats affect incentives for R\&D and optimal mechanism design.
%However, different from our work, the goal of the mechanism design in \citet{tan1992entry} and \citet{piccione1996cost} is cost minimization without focusing on quality enhancement.

\paragraph{Pre-auction investments and endogenous valuations.}
Our work also relates to auctions with pre-bid investments, including \citet{tan1992entry} and \citet{piccione1996cost}, who study how auction formats affect incentives for R\&D and optimal mechanism design.
However, different from our work, the goal of the mechanism design in \citet{tan1992entry} and \citet{piccione1996cost} is cost minimization without focusing on quality enhancement.

A particularly close connection is \citet{gershkov2021endogenous}, who study auctions in which privately informed bidders make costly, unobservable investments before bidding, thereby endogenizing their valuations.
They characterize revenue-maximizing auctions and, under appropriate regularity conditions, establish the optimality of uniform-price and pay-your-bid auctions with suitable reserve prices among symmetric mechanisms.
Our focus differs in both the available incentive instrument and the auctioneer's objective.
In our model, investment affects a stochastic quality observed after allocation, and the auctioneer can reward the winner through a quality-contingent transfer.
Moreover, quality directly contributes to the auctioneer's payoff.
We study the optimal choice of this contract within standard auction formats and establish an asymptotically optimal full-pass-through rule, rather than optimizing the auction's allocation and payment rules alone.

Another closely related antecedent is
\citet{von1994quality}, who studies first-price
procurement auctions with up-front investments in project quality.
As in our model, investments are unobservable and sunk, quality is stochastic, and the principal directly values quality.
In his model, each bidder observes the quality of its project after investing but before submitting a price-quality offer.

\paragraph{Ad and keyword auctions.}
A substantial literature on sponsored-search/keyword (position) auctions incorporates ad \emph{quality} or \emph{relevance} components into the platform’s pricing rules, often through predicted click-through rates or related user-experience signals. 
Existing papers are largely problem-specific, and differ in how quality is modeled and how it enters the platform’s objective.  

A broad literature on sponsored-search and keyword auctions studies how advertiser quality enters ranking and pricing rules, typically through weighted generalized second-price auctions or position-auction variants with given quality parameters~\citep{edelman2007internet,varian2007position,athey2011position,liu2010ex,lahaie2011efficient,thompson2013revenue,li2023optimally,yan2025optimal}.
Another line of work studies how platforms choose the weight placed on quality in allocation and pricing rules, analyzing bid weighting, reserve prices, and efficiency or revenue trade-offs under different quality signals~\citep{liu2010ex,lahaie2011efficient,thompson2013revenue,li2023optimally,yan2025optimal}.
In all these papers, qualities are modeled as \emph{exogenous} variables, which is a key difference from our paper.

Closest to our focus on \emph{endogenous} qualities that are strategically chosen by advertisers, \citet{chen2010keyword} allow advertisers to improve performance at a cost.
In their framework, the auctioneer incentivizes investment by committing to a scoring rule that rewards past performance with better future allocation. 
In contrast, our work focuses on a static setting with sunk costs, where the auctioneer uses an immediate, explicit monetary transfer (a linear contract) to solve the moral hazard problem, rather than relying on long-term dynamic incentives.
There are also many other differences between their model and ours.

\paragraph{Contract theory and moral hazard.}
The problem of incentivizing hidden effort is the core subject of contract theory. 
The classic principal-agent models \citet{ross1973economic,holmstrom1979moral,grossman1992analysis,mirrlees1999theory} establish that contracts can effectively promote better observable outcomes by aligning the agent's incentive to the principal's. 
Recently, there has been significant interest in ``Algorithmic Contract Theory" within the computer science community, focusing on the optimality and the computational complexity of contract design~\citep{holmstrom1987aggregation,hermalin1991moral,dutting2019simple,guruganesh2021contracts,laffont2002theory}. 
Beyond the classical single-principal-single-agent setting, another line of work considers combinatorial~\citep{,dutting2021complexity,dutting2025combinatorial} or multi-agent settings~\citep{babaioff2006combinatorial,feldman2025one} (agents are collaborating in these multi-agent settings, which is different from our setting where agents are competing and only one winner is selected).
A comprehensive survey is provided by \citet{dutting2024algorithmic}.

\subsection{Structure of This Paper}
The formal model of this paper is described in Sec.~\ref{sec:prelim}.
In Sect.~\ref{sec:spa}, we study the second-price auction.
We start with an equilibrium analysis and then find the optimal linear contract.
In Sect.~\ref{sec:equivalence}, we show that the derived optimal contract continues to be optimal in other auctions where symmetric Bayes Nash equilibria exist.
This is shown by proving a revenue equivalence theorem.
In Sect.~\ref{sec:fpa}, we show that the first-price auction admits a symmetric Bayes Nash equilibrium, so the revenue equivalence theorem in Sect.~\ref{sec:equivalence} applies and the contract derived in Sect.~\ref{sec:spa} is also optimal for the first-price auction.
Finally, we conclude our paper in Sect.~\ref{sec:conclusion}.

\section{Preliminaries}
\label{sec:prelim}
We describe our model in the context of ad auctions where the bidders are the advertisers.
As we have mentioned in the introduction, our model captures other applications as well.

Let $[k]=\{1,\ldots,k\}$.
Throughout the paper, superscripts index bidders, and subscripts index action-related terms.
Bidders and action-related terms are indexed by $i$ and $j$, respectively.
Notations defined below are also available in Table~\ref{tab:notations}.

\paragraph{Bidder.}
There is a set $N=[n]$ of bidders, each owning an advertisement, and submitting a non-negative bid to compete for one display slot.
We naturally assume that $n\ge 2$ throughout the paper; some results consider the limit as $n\rightarrow\infty$.
Each bidder adopts an action from an action set $[m]$, which represents the effort level in designing her ad.
An action with a higher effort incurs a higher cost and leads to higher expected quality.
Formally, action $j$ incurs a cost $c_j$.
Ad quality is drawn from a quality set $Q$.
We assume $Q\subseteq\mathbb{R}_{\geq0}$ is bounded.
Each action $j$ leads to a probability distribution $\mathbf{s}_j$ over the qualities $Q$.
The collection $\{\mathbf{s}_j\}_{j\in [m]}$ is common knowledge to all bidders and the auctioneer.
For each action $j$, we denote $\mu_j=\mathbb{E}_{q\sim \mathbf{s}_j}[q]$.
We naturally assume that $0\leq c_1<\cdots<c_m$, and $\mathbf{s}_{j_2}$ stochastically dominates $\mathbf{s}_{j_1}$ and $\mathbf{s}_{j_2}\neq\mathbf{s}_{j_1}$ whenever $j_2>j_1$.
Specifically, we can set $c_1=0$ to indicate there is an option of not participating, or there is a ``default action'' with no cost to resolve the issue of individual rationality.

Each bidder $i$ has a private type $\theta^{(i)}\in\mathbb{R}_{\ge 0}$, which is drawn from distribution $\mathcal{D}^{(i)}$ with support $\Theta^{(i)}$.
We use $D^{(i)}$ and $d^{(i)}$ to denote the cumulative distribution function and the probability density function of $\mathcal{D}^{(i)}$, and we make the following assumption.
\begin{assumption}\label{assumption:distribution}
    Throughout this paper, for each bidder $i$, $\Theta^{(i)}$ is bounded, with $\Theta^{(i)}=[0,B^{(i)}]$, and $\mathcal{D}^{(i)}$ is a continuous distribution with positive probability densities on its support (and so with a strictly increasing $D^{(i)}$).
    When assuming \emph{identical type distribution}, we use $\mathcal{D}$ to refer to the common distribution. Its cumulative distribution function is $D(\cdot)$ and its support is $[0,B]$. We again assume $\mathcal{D}$ is a continuous distribution with the same assumption above.
\end{assumption}
Let $\theta^{(-i)}$ denote the profile of private types excluding bidder $i$'s type $\theta^{(i)}$. 
We assume that each bidder $i$ only knows her own type $\theta^{(i)}$, and holds the beliefs over $\theta^{(-i)}$ according to $\mathcal{D}^{(-i)}=\Pi_{i'\in N\setminus\{i\}}\mathcal{D}^{(i')}$.
The auctioneer observes only all distributions $\{\mathcal{D}^{(i)}\}_{i\in N}$ but not private types.
Each bidder's value for her advertisement is determined by a common valuation function $v:(\bigcup_i\Theta^{(i)})\times Q\rightarrow\mathbb{R}_{\ge 0}$.
The \emph{value} of a bidder with type $\theta$ and realized quality $q$ is denoted by $v(\theta,q)$.
We assume $v$ is \emph{separable}, with $v(\theta,q)=\theta\cdot \vq(q)$ for some strictly increasing bounded function $\vq:Q\rightarrow\mathbb{R}_{> 0}$.\footnote{Ideally, we could set $\vq(0)=0$ to represent $0$ utility of no-participation, if action $1$ with $c_1=0$ means no-participation which always yields a quality of $0$. However, we need positive values of $\vq(\cdot)$ for some technical reasons in Sect.~\ref{sec:fpa}. Readers can take $\vq(0)=\varepsilon$ with a very small $\varepsilon>0$ in this case.}
We denote $\nu_j=\mathbb{E}_{q\sim \mathbf{s}_j}[\vq(q)]$, and we have $\nu_1<\nu_2<\cdots<\nu_m$ as $\vq$ is an increasing function and $\mathbf{s}_{j_2}$ first-order stochastically dominates $\mathbf{s}_{j_1}$ whenever $j_2>j_1$.

It is important to remark that the quality $q\in Q$ of the ad is only observed by the bidders and the auctioneer after the advertising event is over, which is also standard in the contract theory literature.
In practice, the quality $q$ can be measured by users' feedback, click-through rates, and many other factors.
In particular, for every losing bidder of the auction, her ad does not have a chance to be shown, and the quality is never realized.
Before the winner is announced, each bidder with action $j$ can only know the distribution $\mathbf{s}_j$ of the qualities.

\paragraph{On large number of bidders.}
In many parts of our paper, we assume the number of bidders $n$ is large with the assumption $n\to\infty$ or with the asymptotic analysis.
We assume that all the other parameters, $\{c_1,\ldots,c_m\},\{\mathbf{s}_1,\ldots,\mathbf{s}_m\}$, as well as each $\mathcal{D}^{(i)}$ and $B^{(i)}$, do not depend on $n$.

\paragraph{Auction and contract.}
The auctioneer runs an auction $\mathcal{M}=(\mathbf{x},\mathbf{p})$ with allocation rule $\mathbf{x}$ and payment rule $\mathbf{p}=(p^{(1)},\ldots,p^{(n)})$.
The allocation rule $\mathbf{x}$ takes those $n$ bids from the bidders as inputs and outputs a winner from the $n$ bidders.
For each $i\in [n]$, the payment rule $p^{(i)}$ takes those $n$ bids as inputs and outputs the payment for bidder $i$.

We make the following assumption for the auction rule.
\begin{assumption}\label{assumption:auction}
Throughout this paper, $\mathcal{M}=(\mathbf{x},\mathbf{p})$ satisfies the following. 
    \begin{itemize}
    \item The allocation rule $\mathbf{x}$ allocates the ad slot to the bidder with the highest bid. If multiple bidders have the same highest bid, the tie is broken uniformly at random.
    \item For each bidder $i\in [n]$, the payment rule $p^{(i)}$ takes those $n$ bids as input and outputs the payment for bidder $i$. 
    We assume that each $p^{(i)}$ is weakly increasing in bidder $i$'s bid.
    \item Only the winner has a nonnegative payment, and every other bidder pays $0$.
\end{itemize}
\end{assumption}

We will mainly focus on the second-price auction (SPA), denoted by $\mathcal{M}_{\spa}=(\mathbf{x},\mathbf{p}_{\spa}=(p_{\spa}^{(1)},\ldots,p_{\spa}^{(n)}))$ where $\mathbf{x}$ respects Assumption~\ref{assumption:auction} and each $p_{\spa}^{(i)}$ is defined as follows: $p_{\spa}^{(i)}({\tt b}^{(1)},\ldots,{\tt b}^{(n)})=0$ if bidder $i$ loses and $p_{\spa}^{(i)}({\tt b}^{(1)},\ldots,{\tt b}^{(n)})=\max_{j\neq i}{\tt b}^{(j)}$ if bidder $i$ wins.
We will also consider the first-price auction (FPA) $\mathcal{M}_{\fpa}=(\mathbf{x},\mathbf{p}_{\fpa}=(p_{\fpa}^{(1)},\ldots,p_{\fpa}^{(n)}))$ where $\mathbf{x}$ respects Assumption~\ref{assumption:auction} and each $p_{\fpa}^{(i)}$ is defined as follows: $p_{\fpa}^{(i)}({\tt b}^{(1)},\ldots,{\tt b}^{(n)})=0$ if bidder $i$ loses and $p_{\fpa}^{(i)}({\tt b}^{(1)},\ldots,{\tt b}^{(n)})={\tt b}^{(i)}$ if bidder $i$ wins.
We remark that there are many other auctions satisfying Assumption~\ref{assumption:auction}, for example, the \emph{average-between-first-and-second-price-auction} proposed by~\citet{hartman2025s}.

The auctioneer also commits to a contract $T:Q\rightarrow \mathbb{R}_{\ge 0}$, which is a transfer/reward rule to incentivize bidders to take a higher-cost action, thus improving their ad qualities.
If bidder $i$ wins the auction with realized quality $q$, bidder $i$ receives a \emph{transfer} $T(q)$ from the auctioneer.
The transfer is only made to the winning bidder.
We say a contract $T$ is
\begin{itemize}
    \item \emph{monotone}, if $T(q)$ is increasing in $q$.
    \item \emph{linear}, if $T(q)=t\cdot q$ for some constant $t\ge 0$.
\end{itemize}
A linear contract clearly satisfies monotonicity.
Let $\beta_j=\mathbb{E}_{q\sim\mathbf{s}_{j}}[T(q)]$ be the expected transfer given to the winner when action $j$ is played.
The auctioneer will enjoy an extra long-term utility $\Omega:Q\rightarrow \mathbb{R}$ from the quality of the winning ad.
We assume $\Omega(q)=\omega\cdot q$ for some constant $\omega>0$.

When the winner of the auction plays action $j$ with realized quality $q$ and pays $p$, the auctioneer's ex-post revenue is given by $$R=p-T(q)+\Omega(q)=p-T(q)+\omega q.$$
Before the quality is realized, the auctioneer's expected revenue over quality is $$R=p-\E_{q\sim \mathbf{s}_j}[T(q)-\omega q]=p-\beta_j+\omega \mu_j.$$

\begin{table}[t]
    \centering\small
    \begin{tabular}{c|c||c|c}
    \toprule
        $i\in N$ & Index of a bidder & $j\in [m]$ & Index of an action\\
        $\mathcal{D}^{(i)}$ & Type distribution of bidder $i$ & $c_j$ & Cost of action $j$ \\
        $[0,B^{(i)}]$ & support of $\mathcal{D}^{(i)}$ & $\mathbf{s}_j$ & Quality distribution under action $j$ \\
        $\theta^{(i)}$ & Private type of bidder $i$ & $\mu_j:=\mathbb{E}_{q\sim \mathbf{s}_j}[q]$ & Expected quality under action $j$ \\
        $a^{(i)}$ & Action strategy of bidder $i$ & $\nu_j:=\mathbb{E}_{q\sim \mathbf{s}_j}[\vq(q)]$ & $\vq$'s expectation under action $j$ \\
        $b^{(i)}$ & Bidding strategy of bidder $i$ & $T(q)$ & Transfer function \\
        $q\in Q$  & Quality & $\beta_j:=\mathbb{E}_{q\sim \mathbf{s}_j}[T(q)]$ & Expected transfer under action $j$ \\
        $v(\theta,q)=\theta \vq(q)$ & Separable valuation function & $\Omega(q)=\omega q$ & Auctioneer's benefit function from quality \\
    \bottomrule
    \end{tabular}
    \caption{Table of Notations}
    \label{tab:notations}
\end{table}

\subsection{Strategy, Utility, and Equilibrium}
Each bidder's strategy $\sigma^{(i)}=(a^{(i)}, b^{(i)})$ consists of an action strategy $a^{(i)}$, which maps from her private type space to the action space, and a bidding strategy $b^{(i)}$, which maps from her private type space to a nonnegative bid.

Fix a strategy profile $\Sigma=(\sigma^{(1)},\ldots,\sigma^{(n)})$.
Let $\Sigma^{(-i)}$ denote the strategies except bidder $i$.
Given each bidder's type, bidder $i$'s expected utility over quality is defined as
$$u^{(i)}=
\left\{\begin{array}{ll}
\mathbb{E}_{q\sim\mathbf{s}_j}\left[v(\theta^{(i)},q)+T(q)\right]-p-c_j & \mbox{if bidder }i\mbox{ wins the auction} \\
-c_j & \mbox{otherwise}
\end{array}\right.,$$
where $p=p^{(i)}\left(b^{(1)}(\theta^{(1)}),\ldots,b^{(n)}(\theta^{(n)})\right)$ denotes the payment of bidder $i$.
Then, bidder $i$'s \exan expected utility is
$$U^{(i)}(\Sigma) = \E_{\theta^{(i)}\sim\mathcal{D}^{(i)},\theta^{(-i)}\sim \mathcal{D}^{(-i)}}\!\left[
u^{(i)}\!\left(\theta^{(i)},\Sigma(\theta^{(-i)};\theta^{(i)})\right)
\right].$$
Given bidder $i$'s type $\theta^{(i)}$, bidder $i$'s interim expected utility is
$$
U^{(i)}(\theta^{(i)},\Sigma)
=
\E_{\theta^{(-i)}\sim \mathcal{D}^{(-i)}}\!\left[
u^{(i)}\!\left(\theta^{(i)},\Sigma(\theta^{(-i)};\theta^{(i)})\right)
\right].
$$

Below, we define Bayes Nash equilibrium in its interim form.
Throughout our paper, we only consider pure-strategy Bayes Nash equilibrium.
\begin{definition}
 A strategy profile $\Sigma = (\sigma^{(1)}, \ldots, \sigma^{(n)})$ forms a \emph{Bayes Nash equilibrium} if, for every bidder $i\in N$ and every type $\theta^{(i)}\in\Theta^{(i)}$, playing action $a^{(i)}(\theta^{(i)})$ and bidding $b^{(i)}(\theta^{(i)})$ is weakly better than switching to any other pair $(j,{\tt b})$ of action and bid, i.e.,
$$
\E_{\theta^{(-i)}\sim\mathcal{D}^{(-i)}}\left[u^{(i)}\left(\theta^{(i)},(\Sigma^{(-i)}(\theta^{(-i)});\sigma^{(i)}(\theta^{(i)}))\right)\right]
\ge
\E_{\theta^{(-i)}\sim\mathcal{D}^{(-i)}}\left[u^{(i)}\left(\theta^{(i)},(\Sigma^{(-i)}(\theta^{(-i)});(j,{\tt b}))\right)\right]
$$
for any action $j\in [m]$ and any bid ${\tt b}\in\mathbb{R}_{\ge0}$.
\end{definition}

\section{Second-Price Auction: Equilibria Analysis and Revenue Maximization}
\label{sec:spa}
We start by analyzing the second-price auction, which is the easiest to study.
Although we will show the revenue equivalence in the next section (so an optimal contract is largely independent of auction rules), the exact optimal contract parameter is easiest to derive under the second-price auction.
Roughly speaking, we will prove the following results in this section.
\begin{enumerate}
    \item Every bidder has a dominant bidding strategy when fixing an action (Proposition~\ref{prop:naturalbidding}).
    \item Bayes Nash equilibria always exist, and each bidder's action strategy must be of a ``threshold form'' (defined later) in every equilibrium (Theorems~\ref{thm:Bayesian-EquilibriaCharacterization}, \ref{thm:BNE-existence}).
    \item When the bidders' type distributions are identical, there exists a unique symmetric Bayes Nash equilibrium (Theorem~\ref{thm:BNE-identical}).
    \item For identical type distributions, among all linear contracts $T(q)=tq$, as $n\to\infty$, setting $t=\omega$ maximizes the revenue when assuming the unique symmetric Bayes Nash equilibrium is played, i.e., fully passing through the quality value of the auctioneer optimizes the revenue (see Theorem~\ref{thm:BNE-optimalt} and the discussions thereafter);
    \item compared with the setting where a conventional second-price auction without reward is used (i.e., with $t=0$), the use of optimal contracts yields revenue improvements that are almost proportional to the maximum possible gain of the quality score (see the discussions following Theorem~\ref{thm:BNE-optimalt}).
\end{enumerate}

For the ease of illustrating the main ideas, we discuss the case with two actions $m=2$ in this section.
Most of our results extend to the case with general $m$.
Specifically, the main results for the unique existence of Bayes Nash equilibrium and the optimality of setting $t=\omega$ continue to hold for general numbers of actions.
We defer the analysis for general $m$ to Sect.~\ref{sect:Incomplete-m-actions}.

\subsection{Natural Bidding Strategy for Second-Price Auction}
\label{sec:naturalbidding}

\begin{proposition}\label{prop:naturalbidding}
    Under the second-price auction, for each bidder $i$ with type $\theta^{(i)}$ who has decided to play action $j$, the bidding strategy $$b^{(i)}(\theta^{(i)})=\E_{q\sim \mathbf{s}_j}[v(\theta^{(i)},q)+T(q)]=\theta^{(i)}\cdot\nu_j+\beta_j$$ is a dominant bidding strategy under any contract $T$.
\end{proposition}
\begin{proof}
    The reason is similar to the analysis of the dominant strategy in the classical second-price auction.
    Fix an arbitrary bidding profile that contains all $n$ bidders' bids, and in which $b^{(i)}(\theta^{(i)})=\mathbb{E}_{q\sim \mathbf{s}_j}[v(\theta^{(i)},q)+T(q)]$.
    
    Assume that bidder $i$ is the winner, then $u^{(i)}\ge -c_j$.
    Bidding any value higher than the second-highest bid does not change the outcome or the payment, while bidding a value below the second-highest bid makes bidder $i$ lose the auction, and her utility will decrease to $-c_j\le u^{(i)}$.

    On the other hand, assume that bidder $i$ loses the auction by bidding $b^{(i)}(\theta^{(i)})$, then $u^{(i)}=-c_j$.
    Let ${\tt b}_{\max}$ denote the highest bid, then we have ${\tt b}_{\max}\ge \mathbb{E}_{q\sim \mathbf{s}_j}[v(\theta^{(i)},q)+T(q)]$.
    If bidder $i$ deviates to another bid yet still loses the auction, her utility remains $-c_j$,
    If bidder $i$ instead deviates to a bid under which she becomes the winner, her utility becomes $\mathbb{E}_{q\sim \mathbf{s}_j}[v(\theta^{(i)},q)+T(q)]-{\tt b}_{\max}-c_j\le-c_j$ and thus does not increase.
\end{proof}
    
This implies that for each bidder under type $\theta$, every strategy is weakly dominated by a strategy with the same action strategy and $b(\theta)=\mathbb{E}_{q\sim \mathbf{s}_j}[v(\theta,q)+T(q)]$.
Such is called a \emph{natural strategy}.
\begin{definition}
    Considering the second-price auction, for a bidder under type $\theta$, a strategy $\sigma=(a,b)$ is called a \emph{natural strategy} if the bidding strategy $b(\theta)=\mathbb{E}_{q\sim \mathbf{s}_j}[v(\theta,q)+T(q)]$ where $j=a(\theta)$.
\end{definition}

To avoid those well-known unnatural Nash equilibria in second-price auction (e.g., one bidder bids $\infty$ while the others bid $0$), \emph{we will assume bidders only play natural strategies from now on.}
For simplicity, we will only consider a bidder's action when referring to her \emph{strategy} $\sigma^{(i)}$, which implicitly implies the corresponding natural bid.

In the following, when we consider an equilibrium, we restrict it to a natural-strategy equilibrium, that is, always bidding $\mathbb{E}_{q\sim \mathbf{s}_{a(\theta)}}[v(\theta,q)+T(q)]$ under action strategy $a(\cdot)$.
We may write equilibrium instead of natural-strategy equilibrium for short.

\subsection{Equilibrium Analysis}
In this section, we characterize Bayes Nash equilibria.
The technical heart of all theorems in this section is the following lemma, which shows a monotonicity property of a particular bidder when fixing the strategy of the remaining $n-1$ bidders.
Consider a bidder $i$ and fix the strategies $\Sigma^{(-i)}$ of the remaining bidders. 
Let $U_j^{(i)}(\theta^{(i)}, \Sigma^{(-i)})$ be the expected utility of bidder $i$ with type $\theta^{(i)}$ when choosing action $j$. 
We omit $\Sigma^{(-i)}$ and the superscript $(i)$ below, as it is clear from the context.

\begin{lemma}[Key Monotonicity Lemma]\label{prop:monotone_diff}
    Consider the second-price auction and suppose $T(\cdot)$ is monotone.
    For any fixed $\Sigma^{(-i)}$ and any two actions $j_1$ and $j_2$ with $j_1<j_2$, the expected utility difference for a type-$\theta$ bidder $i$, $U_{j_2}(\theta)- U_{j_1}(\theta)$, is increasing in $\theta$.
    In addition, the equation $U_{j_2}(\theta)- U_{j_1}(\theta)=0$ (with the variable $\theta$) has at most one solution.
\end{lemma}
\begin{proof}
    Suppose the remaining $n-1$ bidders play $\Sigma^{(-i)}$.
    Let $F$ be the cumulative distribution function of the random variable describing the highest bid among these $n-1$ bidders.
    Let $f$ be the corresponding probability density function.
    Let $V_{j_1}(\theta;{\tt b}_{\max})$ be the expected utility of bidder $i$ with type $\theta$ when the highest bid among the remaining bidders is fixed to be ${\tt b}_{\max}$.
    Recall that bidder $i$'s natural bidding strategy for action $j_1$ is to bid $\E_{q\sim\mathbf{s}_{j_1}}[v(\theta,q)+T(q)]=\theta\cdot\nu_{j_1}+\beta_{j_1}$.
    We have 
    $$V_{j_1}(\theta;{\tt b}_{\max})=\left\{\begin{array}{ll}
        \theta\cdot\nu_{j_1}+\beta_{j_1}-{\tt b}_{\max}-c_{j_1} & \mbox{if }\theta\cdot\nu_{j_1}+\beta_{j_1}> {\tt b}_{\max}\\
        -c_{j_1} & \mbox{otherwise}
    \end{array}\right..$$
    Let $\beta_{j_2}$ and $V_{j_2}(\theta;{\tt b}_{\max})$ be defined and calculated similarly.
    As a remark, for $\theta\cdot\nu_{j_1}+\beta_{j_1}= {\tt b}_{\max}$, the value of $V_{j_1}(\theta;{\tt b}_{\max})$ is supposed to be somewhere in between $\theta\cdot\nu_{j_1}+\beta_{j_1}-{\tt b}_{\max}-c_{j_1}$ and $-c_{j_1}$ depending on the number of bidders who submit equally highest bids (due to our tie-breaking rule). However, we can arbitrarily define the value of $V_{j_1}(\theta;{\tt b}_{\max})$ since ${\tt b}_{\max}$ is drawn from a continuous distribution $f$ and the probability of $\theta\cdot\nu_{j_1}+\beta_{j_1}= {\tt b}_{\max}$ is $0$. We set it to $-c_{j_1}$.

    We have
    $U_{j_1}(\theta)=\int_{0}^\infty f(x)\cdot V_{j_1}(\theta;x)dx$ and $U_{j_2}(\theta)=\int_{0}^\infty f(x)\cdot V_{j_2}(\theta;x)dx$,
    and
    \begin{equation}\label{eqn:prop-monotone}
    U_{j_2}(\theta)- U_{j_1}(\theta)=\int_0^\infty f(x)\left(V_{j_2}(\theta;x)-V_{j_1}(\theta;x)\right)dx.
    \end{equation}
    To show weak monotonicity of $U_{j_2}(\theta)- U_{j_1}(\theta)$, it suffices to show that $\phi(\theta):=V_{j_2}(\theta;x)-V_{j_1}(\theta;x)$ is increasing in $\theta$ for any \emph{fixed} $x$ (i.e., holds point-wise).

    By monotonicity of $T$ and $\nu_{j_2}>\nu_{j_1}$, for any $\theta>0$, we have   $\theta\cdot\nu_{j_1}+\beta_{j_1}<\theta\cdot\nu_{j_2}+\beta_{j_2},$
    and both sides of the inequalities are increasing in $\theta$.
    By considering three regimes of $\theta$, we have
    $$\phi(\theta)=\left\{\begin{array}{ll}
        -c_{j_2}+c_{j_1} &  \mbox{if }\theta\cdot\nu_{j_1}+\beta_{j_1}<\theta\cdot\nu_{j_2}+\beta_{j_2}\leq x\\
        \theta\cdot\nu_{j_2}+\beta_{j_2}-x-c_{j_2}+c_{j_1} & \mbox{if } \theta\cdot\nu_{j_1}+\beta_{j_1}\leq x<\theta\cdot\nu_{j_2}+\beta_{j_2}\\
        \theta(\nu_{j_2}-\nu_{j_1})+(\beta_{j_2}-\beta_{j_1})-(c_{j_2}-c_{j_1}) & \mbox{if }x<\theta\cdot\nu_{j_1}+\beta_{j_1}<\theta\cdot\nu_{j_2}+\beta_{j_2}
    \end{array}\right..$$
    It is straightforward to check that $\phi(\theta)$ is increasing in $\theta$. Specifically, it is a continuous function with three linear segments: $\phi(\theta)$ is constant in the first regime, then it increases with slope $\nu_{j_2}$ in the second regime, and finally it increases with slope $(\nu_{j_2}-\nu_{j_1})$ in the third regime.
    This proves the first part of the proposition.

    To show the second part of the proposition, if the equation $U_{j_2}(\theta)- U_{j_1}(\theta)=0$ has two different solutions $\theta_1$ and $\theta_2$ with $\theta'<\theta''$, we must have $U_{j_2}(\theta)- U_{j_1}(\theta)=0$ for all $\theta\in[\theta',\theta'']$ due to the weak monotonicity we have shown.
    Since $\phi(\theta)$ is a constant only on the first regime, it must be that the first regime condition $\theta\cdot\nu_{j_1}+\beta_{j_1}<\theta\cdot\nu_{j_2}+\beta_{j_2}\leq x$ happens with probability $1$ for each $\theta\in[\theta',\theta'']$.
    However, we have $U_{j_2}(\theta)- U_{j_1}(\theta)=-c_{j_2}+c_{j_1}<0$, which leads to a contradiction.
    Intuitively, weak monotonicity of $U_{j_2}(\theta)- U_{j_1}(\theta)$ can only happen for those $\theta$'s where bidder $i$ never has a chance to win even if playing action $j_2$; in this case, playing $j_1$ with sunk cost $c_{j_1}$ is strictly better than playing $j_2$ with sunk cost $c_{j_2}$.
\end{proof}

Lemma~\ref{prop:monotone_diff} provides a useful tool to characterize Bayes Nash equilibria, as it characterizes \emph{the best response} for each bidder.
Taking the case with two actions as an example, fixing the strategies of the remaining bidders, bidder $i$'s best response is a \emph{threshold strategy} with the threshold between the two actions $\theta^\ast$ being exactly at the point where $U_1(\theta^\ast)=U_2(\theta^\ast)$.
By (both parts of) Lemma~\ref{prop:monotone_diff}, for $\theta<\theta^\ast$, we have $U_1(\theta)>U_2(\theta)$, so playing action $1$ is better;
Likewise, for $\theta>\theta^\ast$, we have $U_2(\theta)>U_1(\theta)$, so playing action $2$ is better.
Consequently, Lemma~\ref{prop:monotone_diff} implies Theorem~\ref{thm:Bayesian-EquilibriaCharacterization} straightforwardly.

\begin{theorem}[Characterization for BNE under SPA]\label{thm:Bayesian-EquilibriaCharacterization}
    Consider the second-price auction.
    Suppose $T(\cdot)$ is monotone and $m=2$.
    In any Bayes Nash equilibrium, each bidder $i$'s strategy is characterized by a threshold $\theta^{\ast(i)}$ such that action $2$ is played if $\theta^{(i)}>\theta^{*(i)}$ and action $1$ is played if $\theta^{(i)}\leq\theta^{*(i)}$.
\end{theorem}

This theorem can be generalized to any number $m$ of actions.
Instead of a single threshold $\theta^{*(i)}$, each bidder now has up to $m-1$ thresholds that separate the $m$ actions.
See Theorem~\ref{thm:Bayesian-EquilibriaCharacterization-m} in the appendix.

The next theorem shows the existence of Bayes Nash equilibria with two actions.

\begin{theorem}[Existence of BNE under SPA]\label{thm:BNE-existence}
    Consider the second-price auction.
    Suppose $T(\cdot)$ is monotone, a Bayes Nash equilibrium (with the natural bidding strategy) always exists.
\end{theorem}
The proof applies Lemma~\ref{prop:monotone_diff} and uses a potential argument.
The full proof of this theorem is available in Sect.~\ref{sec:BNE-existence-proof}.

In the remaining part of Sect.~\ref{sec:spa}, we will stick to the case with two actions $m=2$ for the ease of illustrating the main ideas.
Our results are generalized to general numbers of actions in Sect.~\ref{sect:Incomplete-m-actions}.

In the next, we assume bidders' type distributions are \emph{identically} $\mathcal{D}$ on $[0,B]$ (following the second part of Assumption~\ref{assumption:distribution}).
We assume linear contract $T(q)=tq$.
We show that a symmetric Bayes Nash equilibrium always exists and it has a simple form described in the next theorem.
Lemma~\ref{prop:monotone_diff} plays an important role again in the proof of the following theorem.

\begin{theorem}[Unique Symmetric BNE for Identical Type Distributions under SPA]\label{thm:BNE-identical}
    Consider the second-price auction.
    When bidders' type distributions are identical (following the second part of Assumption~\ref{assumption:distribution}) and there are two actions $1$ and $2$, upon a linear contract $T(q)=tq$, there is a symmetric Bayes Nash equilibrium where each bidder's strategy is characterized by $\theta^\ast\in[0,B]$ and is defined as
    $$a(\theta)=\left\{\begin{array}{ll}
        \mbox{action }1, & \mbox{if }\theta\leq\theta^* \\
        \mbox{action }2, & \mbox{if }\theta>\theta^*
    \end{array}\right..$$
    where $\theta^\ast$ is the solution to the equation
    \begin{equation}\label{eqn:thetastar}
        (D(\theta^\ast))^{n-1}\left(\theta^\ast(\nu_2-\nu_1)+t(\mu_2-\mu_1)\right)=c_2-c_1
    \end{equation}
    if a solution exists, and $\theta^\ast=B$ otherwise.

    In addition, the equilibrium stated above is the unique symmetric Bayes Nash equilibrium\footnote{Note that another symmetric Bayes Nash equilibrium is to play action $1$ if $\theta<\theta^*$ and action $2$ otherwise if a solution to $\theta^*$ exists. This tie-breaking of action at $\theta^*$ does not affect the expected revenue. Therefore, with a slight abuse of terminology, we say that the equilibrium in the theorem is unique.}.
\end{theorem}
\begin{proof}
Let $\Sigma$ be the strategy profile where every bidder plays the strategy described in the theorem.
Fix $\Sigma^{(-i)}$.
We aim to show that the strategy described in the theorem is a best response to $\Sigma^{(-i)}$.

We write $U_j(\theta)$ for $U_j(\theta,\Sigma^{(-i)})$.
To show that adopting the same threshold $\theta^\ast$ is a best response to $\Sigma^{(-i)}$, we need exactly $U_1(\theta^\ast)=U_2(\theta^\ast)$.
This will imply $U_1(\theta)> U_2(\theta)$ for $\theta<\theta^\ast$ and $U_1(\theta)< U_2(\theta)$ for $\theta>\theta^\ast$ (by Lemma~\ref{prop:monotone_diff}), which implies playing action $1$ is a best response if $\theta
\leq\theta^\ast$ and playing action $2$ is a best response if $\theta\geq\theta^\ast$.

Now we calculate $U_1(\theta^\ast)$ and $U_2(\theta^\ast)$.
For $U_1$ (resp., $U_2$), there is a term $-c_1$ (resp., $-c_2$) regardless of whether or not bidder $i$ wins.
If bidder $i$ loses, this is the only term in $U_1$ (resp., $U_2$).
Otherwise, if bidder $i$ wins, bidder $i$'s type must be the highest, which happens with probability $(D(\theta^\ast))^{n-1}$. 
In this case, bidder $i$ earns an extra expected utility which consists of three terms:
\begin{enumerate}
    \item The expected value: $\theta^\ast\nu_1$ for $U_1$ and $\theta^\ast\nu_2$ for $U_2$.
    \item Minus the expected payment $\bar{p}$, which is the expectation of the highest bids conditioning on all remaining bidders' types are below $\theta^\ast$. Notice that, by definition of $\Sigma^{(-i)}$, all the remaining bidders play action $1$ in this case, and the expected payment is $$\bar{p}=\nu_1\cdot\mathbb{E}_{X_1,\ldots,X_{n-1}\sim\mathcal{D}}[\max\{X_1,\ldots,X_{n-1}\}\mid X_1,\ldots,X_{n-1}\leq\theta^\ast]+t\mu_1.$$
    The exact expression of $\bar{p}$ is not important. What is important is that $\bar{p}$ is the same for both $U_1$ and $U_2$.
    \item Plus the reward, which is $t\mu_1$ for $U_1$ and $t\mu_2$ for $U_2$.
\end{enumerate}
Thus, we have
$$U_1(\theta^\ast)=-c_1+(D(\theta^\ast))^{n-1}\left(\theta^\ast\nu_1-\bar{p}+t\mu_1\right)\quad\mbox{and}\quad U_2(\theta^\ast)=-c_2+(D(\theta^\ast))^{n-1}\left(\theta^\ast\nu_2-\bar{p}+t\mu_2\right).$$
We have $U_2(\theta^\ast)-U_1(\theta^\ast)=0$ if Eqn.~(\ref{eqn:thetastar}) holds (notice that $\bar{p}$ is canceled).

On the other hand, if a solution to Eqn.~(\ref{eqn:thetastar}) does not exist, it implies that 
$$B(\nu_2-\nu_1)+t(\mu_2-\mu_1)< c_2-c_1.$$
Therefore, $U_2(B)<U_1(B)$, implying that playing action $1$ is better for every $\theta\le B$.

To see the uniqueness of $\theta^\ast$, it suffices to see that the solution to (\ref{eqn:thetastar}) is unique, as the left-hand side of the equation is strictly increasing in $\theta^\ast$.
\end{proof}

We remark that if we do not restrict ourselves to \emph{symmetric} Bayes Nash equilibria, there are instances where unnatural asymmetric equilibria exist in which bidders agree on heterogeneous carefully calculated thresholds (see Sect.~\ref{sec:non-uniqueness}).

\subsection{Optimal Contract Design and Advantage of Contracts}
\label{sec:optimalcontract-Bayesian}
In this section, we consider the same setting as in Theorem~\ref{thm:BNE-identical}: identical type distributions, two actions (see Sect.~\ref{sect:Incomplete-m-actions} for general $m$), and linear contract.
We show that setting $t=\omega$ always maximizes the revenue for $n\to\infty$, and we compare the resultant revenue with the benchmark case $t=0$, in which no quality-based reward is applied.

We present the theorem in the limit as $n\to\infty$. 
The computation of the exact optimal parameter $t$ for finite $n$ is deferred to Sect.~\ref{sec:BNE-optimalt-finite-n}, while the proof for the theorem below already captures the essential ideas without complicated calculations.

\begin{theorem}\label{thm:BNE-optimalt}
    Consider the second-price auction.
    When bidders' type distributions are identical (and following Assumption~\ref{assumption:distribution}) and there are two actions $1$ and $2$, consider linear contract $T(q)=tq$ and suppose bidders play the unique symmetric Bayes Nash equilibrium.
    Below, let $$K=\frac{(c_2-c_1)-B(\nu_2-\nu_1)}{\mu_2-\mu_1}.$$
    \begin{itemize}
        \item If $K\geq\omega>0$, setting any $t$ in $[0,K]$ maximizes the expected revenue for $n\to\infty$, and the expected revenue tends to
        $B\nu_1+\omega\mu_1.$
        \item If $\omega>K\geq0$, the optimal $t$ that maximizes the expected revenue satisfies $t\to\omega$ as $n\to\infty$, and the expected revenue tends to
        $$(B\nu_2+\omega\mu_2)-(c_2-c_1)\left(\ln\left(\frac{\omega(\mu_2-\mu_1)+B(\nu_2-\nu_1)}{c_2-c_1}\right)+1\right).$$
        If we set $t=0$ instead, the expected revenue becomes $B\nu_1+\omega\mu_1$ as $n\to\infty$.
        \item If $\omega>0>K$, the optimal $t$ that maximizes the expected revenue satisfies $t\to\omega$ as $n\to\infty$, and the expected revenue tends to
        $$(B\nu_2+\omega\mu_2)-(c_2-c_1)\left(\ln\left(\frac{\omega(\mu_2-\mu_1)+B(\nu_2-\nu_1)}{c_2-c_1}\right)+1\right).$$
        If we set $t=0$ instead, the expected revenue becomes 
        $$B\nu_2+\omega\mu_2-(c_2-c_1)\left(\ln\left(\frac{B(\nu_2-\nu_1)}{c_2-c_1}\right)+1+\frac{\omega(\mu_2-\mu_1)}{B(\nu_2-\nu_1)}\right)$$ 
        as $n\to\infty$.
    \end{itemize}
\end{theorem}

We make a few remarks before proving this theorem.

\paragraph{Interpretation of $K$.}
The value of $K$ that classifies the three cases has a natural meaning: it is the minimum value of the reward to incentivize a hypothetical bidder with maximum type $B$ to play action $2$, supposing this bidder is the only bidder in the auction who always wins (and pays $0$ for having no competitors).
To see this, this bidder's utility for playing action $1$ is $B\nu_1+t\mu_1-c_1$, and her utility for playing action $2$ is $B\nu_2+t\mu_2-c_2$.
It is then straightforward to check that playing action $2$ has a higher utility if and only if $t>K$.

Naturally, lower values of $K$ represent bidders' higher intrinsic incentive to play the more costly action $2$ in the absence of reward.
In particular, if $K<0$, a bidder with a large type may want to play action $2$ even without a contract.

\paragraph{Setting $t\approx\omega$ is optimal.}
As a take-home message, this theorem implies that it is always optimal to set $t\approx\omega$, meaning that the quality reward should be set to approximately the long-term benefit.

The first case corresponds to the scenario where the auctioneer does not care much about long-term quality benefit (with small $\omega$).
In this case, the optimal $t$ is set such that all bidders always play action $1$ in the equilibrium.
The term $t\mu_1$ appears both in the second-highest bidder's bid (natural bidding strategy) and in the revenue loss due to the quality reward to the winner. 
The two appearances of $t\mu_1$ have different signs, and are canceled.
Therefore, $t$ can be set to any sufficiently small values, and these values include $\omega$ (except in the corner case with $K=\omega$, in which case any $t\in[0,\omega]$ is optimal, and we still have $t\approx\omega$ since $\omega$ is at the boundary).

Both the second and the third cases regard large $\omega$.
The second case corresponds to the scenario where bidders' incentive to play action $2$ is small, so that they will play action $1$ for sure if $t=0$.
The third case describes the scenario where bidders have high incentives to play action $2$, and they will play action $2$ for sufficiently large types even if $t=0$.
For both cases, we have the same conclusion that setting $t$ to be approximately $\omega$ is optimal.
However, it is natural to expect that the advantage of using contracts (comparing using $t=\omega$ with setting $t=0$) is more significant in the second case, as we will remark soon.

Given the coincidence $t\to\omega$, one may also wonder if $t=\omega$ (instead of approaching $\omega$) holds for \emph{all} $n$.
We will show in Sect.~\ref{sec:BNE-optimalt-finite-n} that this is not the case.
We will consider the special case with $\mathcal{D}$ being the uniform distribution on $[0,1]$ to demonstrate that the optimal $t$ can be smaller than $\omega$.

As another remark, we have
  $B\nu_1+\omega\mu_1=(B\nu_2+\omega\mu_2)-(c_2-c_1)-(c_2-c_1)\ln\left(\frac{\omega(\mu_2-\mu_1)+B(\nu_2-\nu_1)}{c_2-c_1}\right)$
for $\omega=K$, in which case the revenues in the first two cases in the theorem agree.

\paragraph{Almost linear improvement on quality score.}
Consider the maximum possible expected improvement in the revenue due to high-quality ads, $\omega(\mu_2-\mu_1)$.
Theorem~\ref{thm:BNE-optimalt} implies that, compared with setting $t=0$, using optimal reward factor $t$ improves the revenue by an amount that is almost linear in $\omega(\mu_2-\mu_1)$.
This shows that we can effectively improve the ad quality by using contracts.

For the first case, the auctioneer does not have much incentive for improving the quality (with small $\omega$), in which case we obtain the same revenue with or without contracts.
However, using contracts significantly improves the quality for slightly large $\omega$, as captured by the second and the third cases.

Consider the second case where bidders' own incentives to play action $2$ are low, with $B(\nu_2-\nu_1)\leq(c_2-c_1)$ implied by the pre-condition $\omega>K=\frac{(c_2-c_1)-B(\nu_2-\nu_1)}{\mu_2-\mu_1}\geq0$. By Theorem~\ref{thm:BNE-optimalt}, the increment in revenue is
$(c_2-c_1)\left(W-\ln W-1\right)$,
where we denote $W=\frac{\omega(\mu_2-\mu_1)+B(\nu_2-\nu_1)}{c_2-c_1}$.
Notice that the revenue increment is always positive by the inequality $x-1>\ln x$ for $x>1$ and $W>1$ for $\omega>K$.
Furthermore, by noticing $W\gg\ln W$ for large $W$, the revenue increment is approximately $(c_2-c_1)W$ for large $\omega$, which is $\omega(\mu_2-\mu_1)$ plus the constant $B(\nu_2-\nu_1)$, which is linear in $\omega(\mu_2-\mu_1)$.
Moreover, other than the term $\omega(\mu_2-\mu_1)$, which is the maximum possible improvement in quality score, the auctioneer even has an extra earning of $B(\nu_2-\nu_1)$.
This is due to the fact that the bidders' bids are higher in the natural bidding strategy, which results in higher payment of the winner.

Consider the third case where bidders already have high incentives to play action $2$ even without contracts. 
By Theorem~\ref{thm:BNE-optimalt}, the increment in revenue is $(c_2-c_1)\cdot\left(\frac{\omega(\mu_2-\mu_1)}{B(\nu_2-\nu_1)}-\ln\left(1+\frac{\omega(\mu_2-\mu_1)}{B(\nu_2-\nu_1)}\right)\right)$.
Again, the revenue increment is always positive due to the inequality $x>\ln(1+x)$ (for $x>0$).
By noticing that the linear function grows significantly faster than the logarithmic function, the improvement of the revenue is again approximately a linear factor $\frac{c_2-c_1}{B(\nu_2-\nu_1)}$ of $\omega(\mu_2-\mu_1)$.
This improvement is more significant if $\nu_2-\nu_1$ in the denominator is small.
Intuitively, a large $\nu_2-\nu_1$ implies that bidders' self-motivation of playing action $2$ is high, and, in this case, the effect of a contract becomes less significant.
We should also remark that, even the effect of a contract depends on $\nu_2-\nu_1$, the optimal reward parameter is still $t\to\omega$ as indicated by Theorem~\ref{thm:BNE-optimalt}.
Thus, our result is also somehow surprising.

\begin{proof}[Proof of Theorem~\ref{thm:BNE-optimalt}]
In the following, we assume bidders play the equilibrium described in Theorem~\ref{thm:BNE-identical} with the threshold $\theta^\ast$ described in Eqn.~(\ref{eqn:thetastar}).

We begin by writing the expected revenue in four terms.
Let $N_1$ and $N_2$ be the two random variables describing the bidders with the highest and the second-highest types.
\begin{align*}
    R&=\mathbb{E}[b(N_2)]-t\cdot \mathbb{E}[q^{(N_1)}]+\omega\cdot\mathbb{E}[q^{(N_1)}]\\
    &=\mathbb{E}[\theta^{(N_2)}\vq(q^{(N_2)})]+t\cdot \mathbb{E}[q^{(N_2)}]-t\cdot \mathbb{E}[q^{(N_1)}]+\omega\cdot\mathbb{E}[q^{(N_1)}]\tag{Natural bidding strategy}
\end{align*}

\paragraph{$t$ cannot be infinite.}
We first show that $t$ must be bounded as $n\to\infty$, i.e., we cannot have $t\to\infty$.
This is not immediately clear from the revenue formula above.
Let
$$\gamma(t)=t\cdot (\mathbb{E}[q^{(N_2)}]-\mathbb{E}[q^{(N_1)}])$$
be the sum of the second and third terms in the revenue (which are the only two terms explicitly containing $t$).
Although $(\mathbb{E}[q^{(N_2)}]-\mathbb{E}[q^{(N_1)}])$ is always nonpositive (by Theorem~\ref{thm:BNE-identical}, bidders with higher types are more likely to play the better action $2$), it may also be true that $(\mathbb{E}[q^{(N_2)}]-\mathbb{E}[q^{(N_1)}])\to0$ as $n\to\infty$ (as it is possible that the types of $N_1$ and $N_2$ are both above or below $\theta^\ast$ with high probability for $n\to\infty$).
Thus, $\gamma(t)$ may or may not be bounded from below as $t\to\infty$ at first glance.
However, we will see below that $\gamma(t)\to-\infty$ for $t\to\infty$.

We break down $\gamma(t)$ by computing the two expectations.
Notice that, for two actions, the only case for $\gamma(t)\neq0$ is when $\theta^{(N_1)}>\theta^\ast$ and $\theta^{(N_2)}<\theta^\ast$, in which case two different actions are played by $N_1$ and $N_2$.
This also implies that Eqn.~(\ref{eqn:thetastar}) has a solution $\theta^\ast$ (otherwise, $\theta^\ast=B$ and $\theta^{(N_1)}>\theta^\ast$ is impossible).
This happens precisely when $n-1$ bidders' types are below $\theta^\ast$ and $1$ bidder's type is above $\theta^\ast$.
The probability for this event is therefore $n(D(\theta^\ast))^{n-1}(1-D(\theta^\ast))$.
When this event happens, bidder $N_1$ plays action $2$, with expected quality $\mu_2$, and bidder $N_2$ plays action $1$, with expected quality $\mu_1$.
Thus,
$$\gamma(t)=-t\cdot n(D(\theta^\ast))^{n-1}(1-D(\theta^\ast))(\mu_2-\mu_1).$$
We replace $t(D(\theta^\ast))^{n-1}(\mu_2-\mu_1)$ above by $c_2-c_1-(D(\theta^\ast))^{n-1}\theta^\ast(\nu_2-\nu_1)$ according to Eqn.~(\ref{eqn:thetastar}), and we obtain
\begin{equation}\label{eqn:gamma}
    \gamma(\theta^\ast)=-n(1-D(\theta^\ast))\left(c_2-c_1-(D(\theta^\ast))^{n-1}\theta^\ast(\nu_2-\nu_1)\right).
\end{equation}
Notice that we have replaced $t$ by $\theta^\ast$ for the input of $\gamma(\cdot)$.
We can do this because $\gamma$ can be viewed as a function of $\theta^\ast$ as Eqn.~(\ref{eqn:thetastar}) gives a bijection between $t$ and $\theta^\ast$ (to see it is a bijection, the expression on the left-hand side is strictly increasing in both $t$ and $\theta^\ast$).

Suppose $t\to\infty$.
By Eqn.~(\ref{eqn:thetastar}), we must have $(D(\theta^\ast))^{n-1}\to0$.
The term 
$$\left(c_2-c_1-(D(\theta^\ast))^{n-1}\theta^\ast(\nu_2-\nu_1)\right)$$ 
in $\gamma(\theta^\ast)$ tends to $c_2-c_1$, which is positive.
To avoid $\gamma(t)\to-\infty$, we must then have bounded $n(1-D(\theta^\ast))$.
This implies $1-D(\theta^\ast)$ is upper bounded by $C/n$ for some constant $C$ as $n\to\infty$.
However, if this is the case, $(D(\theta^\ast))^{n-1}$ is then lower-bounded by $e^{-C}$, contradicting the fact that $(D(\theta^\ast))^{n-1}\to0$.
Therefore, $t$ is bounded as $n\to\infty$.

Furthermore, we must have $D(\theta^\ast)\to1$, for otherwise, we have $(D(\theta^\ast))^{n-1}\to0$ 
and Eqn.~(\ref{eqn:thetastar}) fails for bounded $t$.
On the other hand, $D(\theta^\ast)\not\to1$ implies $\theta^\ast\not\to B$ (since $D$ is strictly increasing), which implies $\theta^\ast$ must be a solution of Eqn.~(\ref{eqn:thetastar}) by Theorem~\ref{thm:BNE-identical}, leading to a contradiction.

Since $\mathcal{D}$ is a continuous distribution with strictly increasing $D(\cdot)$, $D(\theta^\ast)\to1$ also implies $\theta^\ast\to B$.
Equation~(\ref{eqn:thetastar}) then becomes
\begin{equation}\label{eqn:thetastar'}
        (D(\theta^\ast))^{n-1}\left(B(\nu_2-\nu_1)+t(\mu_2-\mu_1)\right)=c_2-c_1
\end{equation}
It has solutions if and only if $t\geq\frac{(c_2-c_1)-B(\nu_2-\nu_1)}{\mu_2-\mu_1}$.
Note that $\theta^*=B$ when $t=\frac{(c_2-c_1)-B(\nu_2-\nu_1)}{\mu_2-\mu_1}$.
On the other hand, if $t<\frac{(c_2-c_1)-B(\nu_2-\nu_1)}{\mu_2-\mu_1}$, Eqn.~(\ref{eqn:thetastar'}) does not have a solution.
In this case, we have $\theta^\ast=B$ by Theorem~\ref{thm:BNE-identical} and all bidders will play action $1$ for sure.

\paragraph{Finding optimal $t$.}
We explore two options: $t\le \frac{(c_2-c_1)-B(\nu_2-\nu_1)}{\mu_2-\mu_1}$ and $t>\frac{(c_2-c_1)-B(\nu_2-\nu_1)}{\mu_2-\mu_1}$.
For the first option, all bidders will surely play action $1$.
The expected revenue is
\begin{align*}
    R&=\mathbb{E}[\theta^{(N_2)}\vq(q^{(N_2)})]+t\cdot \mathbb{E}[q^{(N_2)}]-t\cdot \mathbb{E}[q^{(N_1)}]+\omega\cdot\mathbb{E}[q^{(N_1)}]\\
    &=\mathbb{E}[\theta^{(N_2)}]\cdot\nu_1+t\mu_1-t\mu_1+\omega\mu_1\\
    &=B\nu_1+\omega\mu_1\tag{since $\mathbb{E}[\theta^{(N_2)}]\to B$ for $n\to\infty$}
\end{align*}
We conclude that setting any $t$ with $t\in[0,\frac{(c_2-c_1)-B(\nu_2-\nu_1)}{\mu_2-\mu_1}]$ results in the same revenue
\begin{equation}\label{eqn:Bayesian-revenue-action1}
    B\nu_1+\omega\mu_1.
\end{equation}

In the remaining part, we will explore the option $t>\frac{(c_2-c_1)-B(\nu_2-\nu_1)}{\mu_2-\mu_1}$ and optimize $t$ in this range, in which case $\theta^\ast$ and $t$ are related by Eqn.~(\ref{eqn:thetastar'}). 
Keep in mind that we also have the option for $t=0$, and we need to compare the two options at the end.

To work out the optimal $t$, we have seen that it is equivalent to working out the optimal $\theta^\ast$. Let $D(\theta^\ast)=1-\frac{x}{n}$.
Since there is a bijection between $x$ and $\theta^\ast$, we will work out the optimal $x$ instead of $\theta^\ast$.
As we have proved $D(\theta^\ast)\to1$, we must have $\frac{x}{n}\to0$.
As a result, we have $(D(\theta^\ast))^{n}\to e^{-x}$, and the same is for $(D(\theta^\ast))^{n-1}$.

We compute the four terms in the expected revenue.

We begin by computing the first term $\mathbb{E}[\theta^{(N_2)}\vq(q^{(N_2)})]$.
Consider $n$ independent random variables with distribution $\mathcal{D}$, and consider the distribution of the second largest random variable.
Let $\pi$ be the probability density function for this distribution.
We have
$$\pi(z)=n(n-1)d(z)(1-D(z))(D(z))^{n-2},$$
where $D$ and $d$ are the cumulative distribution function and probability density function for $\mathcal{D}$.
The expectation for the first term of the revenue is
$$\int_0^{\theta^\ast}\pi(z)z\nu_1 dz+\int_{\theta^\ast}^B\pi(z)z\nu_2 dz=\nu_1\int_0^B\pi(z)zdz+(\nu_2-\nu_1)\int_{\theta^\ast}^B\pi(z)zdz.$$
The first integral $\int_0^B\pi(z)zdz$ is just the expectation of the second largest type, which approaches to $B$ as $n\to\infty$.
For the second integral, the lower integral limit $\theta^\ast$ approaches to the upper limit $B$, and we can approximate $z$ by $B$ in the integrand.
Thus, for $n\to\infty$, we have
\begin{align*}
    \int_{\theta^\ast}^B\pi(z)zdz&=B\int_{\theta^\ast}^Bn(n-1)d(z)(1-D(z))(D(z))^{n-2}dz\\
    &=B\left[n(D(z))^{n-1}-(n-1)(D(z))^n\right]_{\theta^\ast}^B\\
    &=B\left(1-n(D(\theta^\ast))^{n-1}+(n-1)(D(\theta^\ast))^n\right)\\
    &=B\left(1-(D(\theta^\ast))^{n}-n(D(\theta^\ast))^{n-1}(1-D(\theta^\ast))\right)
\end{align*}
By substituting $D(\theta^\ast)=1-\frac xn$ and substituting $(D(\theta^\ast))^{n}$ and $(D(\theta^\ast))^{n}$ by $e^{-x}$ to the above equation, the first term of the revenue is
$$R_1(x)=\nu_1\cdot B+(\nu_2-\nu_1)\cdot B(1-e^{-x}-xe^{-x})=B((e^{-x}
+xe^{-x})\nu_1+(1-e^{-x}-xe^{-x})\nu_2).$$

Similarly, for the second and the third terms, by using the same substitutions as (\ref{eqn:gamma}), we have
$$R_{23}(x)=-x(c_2-c_1-B(\nu_2-\nu_1)e^{-x}).$$
For the fourth term,
$$R_4(x)=\omega\cdot\mathbb{E}[q^{(N_1)}]=\omega\cdot\left((D(\theta^\ast))^n\mu_1+(1-(D(\theta^\ast))^n)\mu_2\right)=\omega(e^{-x}\mu_1+\mu_2-e^{-x}\mu_2).$$
Putting $R_1,R_{23}$, and $R_4$ together and simplifying, we have
\begin{equation}\label{eqn:Bayesian-revenue}
    R(x)=(B\nu_2+\mu_2\omega)-x(c_2-c_1)-e^{-x}\omega(\mu_2-\mu_1)-e^{-x}B(\nu_2-\nu_1).
\end{equation}
To find $x$ that optimizes $R$, we compute the derivatives:
$$R'(x)=-(c_2-c_1)+e^{-x}\omega(\mu_2-\mu_1)+e^{-x}B(\nu_2-\nu_1).$$
Since
$$R''(x)=-e^{-x}\omega(\mu_2-\mu_1)-e^{-x}B(\nu_2-\nu_1)<0,$$
$R$ is maximized when $R'(x)=0$.
This yields
$$e^{-x}=\frac{c_2-c_1}{\omega(\mu_2-\mu_1)+B(\nu_2-\nu_1)}.$$

Notice that $x$ must not be negative; we need to consider two cases:
\begin{itemize}
    \item Case 1: $\frac{c_2-c_1}{\omega(\mu_2-\mu_1)+B(\nu_2-\nu_1)}\geq 1$, in which case $\omega\leq\frac{(c_2-c_1)-B(\nu_2-\nu_1)}{\mu_2-\mu_1}$, and $R$ is optimized at $x=0$.
    In this case, we have $D(\theta^\ast)=1$, so $\theta^\ast=B$. 
    This goes back to the first option $t\le\frac{(c_2-c_1)-B(\nu_2-\nu_1)}{\mu_2-\mu_1}$ we have explored before, with revenue in Eqn.~(\ref{eqn:Bayesian-revenue-action1}).
    In this case, setting $t$ to be any value between $0$ and $\frac{(c_2-c_1)-B(\nu_2-\nu_1)}{\mu_2-\mu_1}$ is optimal.
    Since $\omega<\frac{(c_2-c_1)-B(\nu_2-\nu_1)}{\mu_2-\mu_1}$, we can set $t=\omega$, and obtain the optimal revenue $B\nu_1+\omega\mu_1.$
    \item Case 2: $\frac{c_2-c_1}{\omega(\mu_2-\mu_1)+B(\nu_2-\nu_1)}<1$, in which case $\omega>\frac{(c_2-c_1)-B(\nu_2-\nu_1)}{\mu_2-\mu_1}$, and $R$ is optimized at $x=\ln\left(\frac{\omega(\mu_2-\mu_1)+B(\nu_2-\nu_1)}{c_2-c_1}\right)$. Since $(D(\theta^\ast))^{n-1}=e^{-x}$, substitute the above expression into Eqn.~(\ref{eqn:thetastar'}) and we have
    $$\frac{c_2-c_1}{\omega(\mu_2-\mu_1)+B(\nu_2-\nu_1)}\left(B(\nu_2-\nu_1)+t(\mu_2-\mu_1)\right)=c_2-c_1.$$
    Thus, we have $t=\omega$.
    The corresponding revenue is
    $$(B\nu_2+\omega\mu_2)-(c_2-c_1)-(c_2-c_1)\ln\left(\frac{\omega(\mu_2-\mu_1)+B(\nu_2-\nu_1)}{c_2-c_1}\right).$$
\end{itemize}
This verifies that the optimal choice of $t$ in the three cases of the theorem, as well as the corresponding optimal expected revenue.
It remains to verify the expected revenue for setting no reward $t=0$.

For the second case $\omega>\frac{(c_2-c_1)-B(\nu_2-\nu_1)}{\mu_2-\mu_1}\geq0$, we have $B(\nu_2-\nu_1)\leq c_2-c_1$.
For $t=0$, Eqn.~(\ref{eqn:thetastar'}) only has solution $D(\theta^\ast)=1$ for $B(\nu_2-\nu_1)=c_2-c_1$, and it has no solution if $B(\nu_2-\nu_1)< c_2-c_1$.
In both cases, we have $\theta^\ast=B$, and all bidders play action $1$.
We have seen in (\ref{eqn:Bayesian-revenue-action1}) the expected revenue is $B\nu_1+\omega\mu_1$.

For the third case $\omega>0>\frac{(c_2-c_1)-B(\nu_2-\nu_1)}{\mu_2-\mu_1}$, Eqn.~(\ref{eqn:thetastar'}) implies $(D(\theta^\ast))^{n-1}=\frac{c_2-c_1}{B(\nu_2-\nu_1)}$ for $t=0$.
This implies $e^{-x}=\frac{c_2-c_1}{B(\nu_2-\nu_1)}$ for $D(\theta^\ast)=1-\frac xn$.
By Eqn.~(\ref{eqn:Bayesian-revenue}), we have
$$R(x)=B\nu_2+\omega\mu_2-(c_2-c_1)\left(\ln\left(\frac{B(\nu_2-\nu_1)}{c_2-c_1}\right)+1+\frac{\omega(\mu_2-\mu_1)}{B(\nu_2-\nu_1)}\right),$$
which concludes the theorem.
\end{proof}

\section{Other Auction Rules and Revenue Equivalence}
\label{sec:equivalence}
In the previous section, we have analyzed the equilibria for the second-price auction.
We show that, when bidders' types are identically distributed, there exists a unique symmetric equilibrium and fully passing through the quality value to the winner (i.e., setting $t\approx\omega$) optimizes the revenue.
In this section, we will show that this holds for other auction rules as well.
In particular, we derive the following revenue equivalence theorem which says that the auctioneer's expected revenue is the same in all auctions where symmetric equilibria exist.

\begin{theorem}\label{thm:revenue-equivalence}
    Consider any $n$ and any auction rule $\mathcal{M}=(\mathbf{x},\mathbf{p})$ satisfying Assumption~\ref{assumption:auction}.
    Suppose bidders' types are identically distributed (Assumption~\ref{assumption:distribution}).
    Suppose a linear contract $T(q)=tq$ is used and there exists a symmetric Bayes Nash equilibrium $\Sigma=(\sigma^{(1)},\ldots,\sigma^{(n)})$ where, for each $i\in[n]$, $\sigma^{(i)}=(a,b)$ with strictly increasing bidding strategy $b:\Theta\to\mathbb{R}_{\geq0}$.
    The expected revenue obtained in this auction is equal to the expected revenue obtained in $\mathcal{M}_{\spa}$ under the same contract $T(q)=tq$, when bidders play the unique symmetric Bayes Nash equilibrium $\Sigma_{\spa}$ characterized in Theorem~\ref{thm:BNE-identical}.
\end{theorem}

We prove this theorem for two actions in this section, and the extension to $m$ actions is available in Appendix~\ref{sec:revenue-equivalence-m}.

We first prove the following proposition which will be used in the proof of Theorem~\ref{thm:revenue-equivalence} and also in Sect.~\ref{sec:fpa}.
The statement of the proposition below is similar to that of Theorem~\ref{thm:BNE-identical}.
However, the proposition below only says that threshold action strategy is a \emph{necessary condition} for $(a,b)$ being a symmetric Bayes Nash equilibrium, while Theorem~\ref{thm:BNE-identical} is a \emph{characterization} of symmetric Bayes Nash equilibria.
The necessary condition for $a$ in the proposition below is due to the strict monotonicity of $b$: we consider fixing everyone's bidding strategy to be $b$ and consider bidder $i$'s deviation of the action strategy; the crucial observation is that bidder $i$ with type $\theta^{(i)}$ wins if and only if $\theta^{(i)}$ is the highest (due to strict monotonicity and symmetry of the bidding strategy).
On the other hand, for Theorem~\ref{thm:BNE-identical} regarding the second-price auction, the bidding strategy is coupled with the action strategy (the natural bidding strategy in Sect.~\ref{sec:naturalbidding}), so we cannot analyze the equilibria by fixing bidding strategies while changing action strategies.
The proof of Theorem~\ref{thm:BNE-identical} is based on the threshold characterizations in Lemma~\ref{prop:monotone_diff} and Theorem~\ref{thm:Bayesian-EquilibriaCharacterization}.
In particular, Lemma~\ref{prop:monotone_diff} and Theorem~\ref{thm:Bayesian-EquilibriaCharacterization} assume the natural bidding strategy and only work for the second-price auction.

\begin{proposition}\label{prop:strictlyIncreasingBidding->threshold}
    Consider identical bidders' type distributions, $m=2$, and any linear contract $T(q)=tq$.
    Suppose there exists a symmetric Bayes Nash equilibrium where everyone plays $(a,b)$ with strictly increasing $b$.
    There must exist $\theta^\ast\in[0,B]$ such that
    $$a(\theta)=\left\{\begin{array}{ll}
        \mbox{action }1, & \mbox{if }\theta\leq\theta^* \\
        \mbox{action }2, & \mbox{if }\theta>\theta^*
    \end{array}\right..$$
    where $\theta^\ast$ is the solution to the equation
    \begin{equation}\label{eqn:thetastar-increasingbidding}
        (D(\theta^\ast))^{n-1}\left(\theta^\ast(\nu_2-\nu_1)+t(\mu_2-\mu_1)\right)=c_2-c_1
    \end{equation}
    if a solution exists, and $\theta^\ast=B$ otherwise.
\end{proposition}
\begin{proof}
    Consider any bidder $i$ with type $\theta$ (we have omitted the superscript $(i)$ for notation simplicity).
    Since $\Sigma$ is a symmetric equilibrium and $b$ is strictly increasing, bidder $i$ wins with probability $D(\theta)^{n-1}$ (i.e., if and only if her type is the highest).
    Let $P(\theta)=\mathbb{E}_{\theta^{(-i)}\sim\mathcal{D}^{-i}}[p^{(i)}(b(\theta^{(1)}),\ldots,b(\theta^{(n)}))]$ be the expected payment of bidder $i$. 
    Notice that $P$ is a function from $\Theta$ (bidder $i$'s type) to $\mathbb{R}_{\geq0}$, and, in particular, $P$ is the ``personal'' expected payment function for bidder $i$ under the strategy profile $\Sigma$.

    The expected utility for bidder $i$ when playing action $1$ is $D(\theta)^{n-1}(\theta\nu_1+t\mu_1)-P(\theta)-c_1$, and the expected utility for playing action $2$ is $D(\theta)^{n-1}(\theta\nu_2+t\mu_2)-P(\theta)-c_2$.
    For $(a,b)$ to be a best response to $\Sigma^{(-i)}$, it must be that action $2$ is played if and only if
    $$D(\theta)^{n-1}(\theta\nu_1+t\mu_1)-P(\theta)-c_1\leq D(\theta)^{n-1}(\theta\nu_2+t\mu_2)-P(\theta)-c_2,$$
    or equivalently,
    $$D(\theta)^{n-1}\left(\theta(\nu_2-\nu_1)+t(\mu_2-\mu_1)\right)\geq c_2-c_1.$$
    Since the left-hand side is strictly increasing in $\theta$ (as we have assumed $D(\cdot)$ is strictly increasing), $a$ must be the threshold strategy with the threshold $\theta^\ast$ exactly as described in Theorem~\ref{thm:BNE-identical}.
\end{proof}

Now we are ready to prove Theorem~\ref{thm:revenue-equivalence}.
\begin{proof}[Proof of Theorem~\ref{thm:revenue-equivalence}]
    It suffices to prove the following two claims: 
    \begin{enumerate}
        \item the action strategy $a$ must be the same as it is in $\Sigma_{\spa}$, and
        \item for each bidder $i$ and each $\theta^{(i)}\in\Theta$, the expected payment (where the expectation is taken over $\theta^{(-i)}\sim\mathcal{D}^{(-i)}$) of bidder $i$ under $\Sigma$ in $\mathcal{M}$ is the same as the expected payment under $\Sigma_{\spa}$ in $\mathcal{M}_{\spa}$.
    \end{enumerate}
    Recall that the auctioneer's revenue consists of three parts: the payment from the winner, minus the quality reward to the winner, and plus the quality value obtained.
    When claim 2 holds, the expectation of the first part is the same as it is in $\mathcal{M}_{\spa}$; when claim 1 holds, the expectations of the second and the third parts are the same as they are in $\mathcal{M}_{\spa}$.

    Consider any bidder $i$ with type $\theta$ (we have omitted the superscript $(i)$ for notation simplicity).
    Since $\Sigma$ is a symmetric equilibrium and $b$ is strictly increasing, bidder $i$ wins with probability $D(\theta)^{n-1}$ (i.e., if and only if her type is the highest).
    Let $P(\theta)=\mathbb{E}_{\theta^{(-i)}\sim\mathcal{D}^{-i}}[p^{(i)}(b(\theta^{(1)}),\ldots,b(\theta^{(n)}))]$ be the expected payment of bidder $i$. 
    Notice that $P$ is a function from $\Theta$ (bidder $i$'s type) to $\mathbb{R}_{\geq0}$, and, in particular, $P$ is the ``personal'' expected payment function for bidder $i$ under the strategy profile $\Sigma$.

    Claim 1 follows immediately by comparing Proposition~\ref{prop:strictlyIncreasingBidding->threshold} and Theorem~\ref{thm:BNE-identical}.

    To prove claim 2, we first consider $\theta_1,\theta_2\in[0,\theta^\ast]$ with $\theta_1<\theta_2$.
    If bidder $i$'s type is $\theta_1$, she will play action $1$ and bid $b(\theta_1)$ in the equilibrium, and the expected utility is $D(\theta_1)^{n-1}(\theta_1\nu_1+t\mu_1)-P(\theta_1)-c_1$.
    Suppose she bids $b(\theta_2)$ instead (while still playing action $1$), her expected utility is $D(\theta_2)^{n-1}(\theta_1\nu_1+t\mu_1)-P(\theta_2)-c_1$. In this case, bidder $i$ wins with probability $D(\theta_2)^{n-1}$ instead, and her expected payment becomes $P(\theta_2)$.
    Since bidder $i$ should have played a best response in the equilibrium, we must have
    $$D(\theta_1)^{n-1}(\theta_1\nu_1+t\mu_1)-P(\theta_1)-c_1\geq D(\theta_2)^{n-1}(\theta_1\nu_1+t\mu_1)-P(\theta_2)-c_1,$$
    and equivalently,
    $$P(\theta_2)-P(\theta_1)\geq\left(D(\theta_2)^{n-1}-D(\theta_1)^{n-1}\right)(\theta_1\nu_1+t\mu_1).$$
    Now, suppose bidder $i$'s true type is $\theta_2$ instead, and consider the deviation when she bids $b(\theta_1)$.
    We can obtain
    $$D(\theta_2)^{n-1}(\theta_2\nu_1+t\mu_1)-P(\theta_2)-c_1\geq D(\theta_1)^{n-1}(\theta_2\nu_1+t\mu_1)-P(\theta_1)-c_1,$$
    and equivalently, 
    $$P(\theta_2)-P(\theta_1)\leq\left(D(\theta_2)^{n-1}-D(\theta_1)^{n-1}\right)(\theta_2\nu_1+t\mu_1).$$
    Putting together, we have
    $$\left(D(\theta_2)^{n-1}-D(\theta_1)^{n-1}\right)(\theta_1\nu_1+t\mu_1)\leq P(\theta_2)-P(\theta_1)\leq\left(D(\theta_2)^{n-1}-D(\theta_1)^{n-1}\right)(\theta_2\nu_1+t\mu_1).$$
    Since $D(\cdot)$ is continuous, by letting $\theta_2$ approach $\theta_1$ and denoting $z=\theta_1$, we see that $P(\cdot)$ must be continuous on $[0,\theta^\ast]$, and we have
    $$d(P(z))=d\left(D(z)^{n-1}\right)\cdot (z\nu_1+t\mu_1),$$
    where $d(\cdot)$ is the differential operator.\footnote{We use the differential operator in this proof for simplicity and clarity. It is straightforward to convert the proof into a more formal language with Riemann sums.}
    By integrating $z$ over $[0,\theta]$ for $\theta\in[0,\theta^\ast)$, we have
    $$P(\theta)=P(0)+\int_{0}^{\theta}(z\nu_1+t\mu_1)\,d\left(D(z)^{n-1}\right).$$
    Since a losing bidder pays $0$ and a bidder with type $0$ loses with probability $1$, we have $P(0)=0$, so
    $$P(\theta)=\int_{0}^{\theta}(z\nu_1+t\mu_1)\,d\left(D(z)^{n-1}\right).$$
    This is exactly the expected payment when the second-price auction is used.

    By a similar analysis for $\theta\in[\theta^\ast,B]$, we can prove that
$$P(\theta)=P(\theta^\ast)+\int_{\theta^\ast}^\theta(z\nu_2+t\mu_2)\,d\left(D(z)^{n-1}\right)=\int_{0}^{\theta^\ast}(z\nu_1+t\mu_1)\,d\left(D(z)^{n-1}\right)+\int_{\theta^\ast}^\theta(z\nu_2+t\mu_2)\,d\left(D(z)^{n-1}\right).$$
    Again, this is exactly the expected payment when the second-price auction is used.
    We have proved claim 2.

    As a remark, the proof implicitly uses the fact that $P$ is continuous on $[0,B]$.
    We have shown that $P$ is continuous on $[0,\theta^\ast]$ (in particular, it is left-continuous on $\theta^\ast$), and a similar argument can show that $P$ is continous on $[\theta^\ast,B]$.
    Thus, it is continuous on $[0,B]$.
\end{proof}

\section{First-Price Auction}
\label{sec:fpa}
In this section, we show that, under the first-price auction, a symmetric Bayes Nash equilibrium always exists for identical type distributions, in which the bidding strategy is strictly increasing in the bidder's type.
We show the existence for two actions in the following, and the extension to $m$ actions in Appendix~\ref{append:fpa}.
Combined with the revenue equivalence result established in the previous section, this implies that setting $t\approx\omega$ asymptotically maximizes the revenue under the symmetric first-price auction equilibrium as $n\to\infty$.

\begin{theorem}[Existence of Symmetric BNE for Identical Type Distributions under FPA]\label{thm:BNE-identical-fpa}
    Consider the first-price auction when bidders' type distributions are identical (following Assumption~\ref{assumption:distribution}). 
    Suppose a linear contract $T(q)=tq$ is used and there are two actions $1$ and $2$.
    Then there exists a symmetric Bayes Nash equilibrium $\Sigma=(\sigma^{(1)},\ldots,\sigma^{(n)})$ where every bidder $i\in[n]$ adopts the same strategy $\sigma^{(i)}=(a,b)$.
    The action strategy $a$ is a threshold strategy with threshold $\theta^\ast\in[0,B]$:
$$
a(\theta)=
\begin{cases}
\text{action }1, & \theta\le\theta^\ast,\\
\text{action }2, & \theta>\theta^\ast.
\end{cases}
$$
    where $\theta^\ast$ is the solution to the equation
    \begin{equation}
        (D(\theta^\ast))^{n-1}\left(\theta^\ast(\nu_2-\nu_1)+t(\mu_2-\mu_1)\right)=c_2-c_1\label{eqn:fpa-threshold}
    \end{equation}
    if a solution exists, and $\theta^\ast=B$ otherwise.

    Each bidder adopts a continuous and strictly increasing bidding strategy $b(\theta)$, which is given by
    $$
    b(\theta)
    =
    \theta\nu_{a(\theta)}+t\mu_{a(\theta)}
    -
    \frac{
        c_{a(\theta)}-c_1
        +
        \displaystyle\int_0^\theta
        D(z)^{n-1}\nu_{a(z)}\,dz
    }{
        D(\theta)^{n-1}
    }
    $$
    for every $\theta>0$, and set $b(0):=t\mu_1$ for $\theta=0$.
    Equivalently, when $\theta\le\theta^\ast$,
    \begin{equation}\label{eqn:fpa-1}
         b(\theta)
        =
        \theta\nu_1+t\mu_1
        -
        \frac{
            \displaystyle\int_0^\theta D(z)^{n-1}\nu_1\,dz
        }{
            D(\theta)^{n-1}
        },
    \end{equation}
    and when $\theta>\theta^\ast$,
    \begin{equation}\label{eqn:fpa-2}
         b(\theta)
        =
        \theta\nu_2+t\mu_2
        -
        \frac{
            c_2-c_1
            +
            \displaystyle\int_0^{\theta^\ast}D(z)^{n-1}\nu_1\,dz
            +
            \displaystyle\int_{\theta^\ast}^{\theta}D(z)^{n-1}\nu_2\,dz
        }{
            D(\theta)^{n-1}
        }.
    \end{equation}
\end{theorem}
\begin{proof}
    According to Proposition~\ref{prop:strictlyIncreasingBidding->threshold}, if a symmetric equilibrium exists with a continuous and strictly increasing bidding strategy, then the action strategy in the equilibrium must be the same as in $\Sigma_\spa$.
    Therefore, it suffices to prove the following three claims:  
    \begin{enumerate}
        \item If a symmetric equilibrium with a strictly increasing bidding strategy exists, it must be as described in the theorem.
        We prove this by showing that the expected utility of each bidder is uniquely determined among all such possible equilibria. 
        The proof also provides an explicit construction of the bidding function.
        \item The bidding strategy is continuous and strictly increasing.
        \item For a bidder with fixed type $\theta\in\Theta$, the deviation to any pair of action and bid is not profitable. This shows that the symmetric profile described in the theorem is a Bayes Nash equilibrium.
    \end{enumerate}
    
    For claim 1, consider any symmetric equilibrium with the action strategy $a(\cdot)$ described above and some strictly increasing bidding strategy $b(\cdot)$.
    For a bidder with type $\theta$, she wins the auction when she has the highest type with probability $D(\theta)^{n-1}$.
    Let $\calU(\theta)$ denote the expected utility when she follows the prescribed strategy, we have
    \begin{equation}
    \calU(\theta) = D(\theta)^{n-1}
        \left(
            \theta\nu_{a(\theta)}
            +t\mu_{a(\theta)}
            -b(\theta)
        \right)
        -
        c_{a(\theta)}.
        \label{eqn:fpa-utility}
    \end{equation}

    We will show a necessary condition of $\calU(\cdot)$ under any symmetric equilibrium with a strictly increasing bidding strategy.
    Consider two types $\theta_1>\theta_2$ in $[0,B]$ that satisfy $a(\theta_1)=a(\theta_2)=j$. 
    Under the equilibrium, bidder with type $\theta_1$ will not deviate to bid $b(\theta_2)$, therefore it is necessary to guarantee
    $$\calU(\theta_1)-D(\theta_2)^{n-1}\left(\theta_1\nu_j+t\mu_j-b(\theta_2)\right)+c_j\ge 0.$$
    Substituting $\calU(\theta_2)=D(\theta_2)^{n-1}(\theta_2\nu_j+t\mu_j-b(\theta_2))-c_j$, the above inequality is equivalent to
    $$\calU(\theta_1)-\calU(\theta_2)\ge D(\theta_2)^{n-1}(\theta_1-\theta_2)\nu_j.$$
    
    Similarly, bidder with type $\theta_2$ will not deviate to bid $b(\theta_1)$, which gives us
    $$\calU(\theta_1)-\calU(\theta_2)\le D(\theta_1)^{n-1}(\theta_1-\theta_2)\nu_j.$$
    
    Dividing the two inequalities by $\theta_1-\theta_2$ gives
    $$
        D(\theta_2)^{n-1}\nu_j
        \le
        \frac{\calU(\theta_1)-\calU(\theta_2)}{\theta_1-\theta_2}
        \le
        D(\theta_1)^{n-1}\nu_j.
    $$
    Fix a point $\theta\neq\theta^*$, then the action strategy $a(\cdot)$ is locally constant around $\theta$. 
    By taking $\theta_1\rightarrow\theta$ and $\theta_2\rightarrow\theta$, the left and right bounds converge to the same value as $D(\cdot)^{n-1}$ is continuous and strictly increasing. 
    Therefore, at every point except $\theta^*$, $\calU$ is
    differentiable with
    $$
        \mathcal U'(\theta)=D(\theta)^{n-1}\nu_{a(\theta)}.
    $$
    The only point at which the action strategy may fail to be locally constant is the threshold $\theta^\ast$, which is a single point and does not affect the subsequent integration.
    
    The lowest type $0$ wins with probability zero and chooses the lowest-cost action, so
    $$
        \calU(0)=-c_1.
    $$
    Integrating, we have
    \begin{equation}
        \calU(\theta)
        =
        -c_1+\int_0^\theta D(z)^{n-1}\nu_{a(z)}\,dz,
        \label{eqn:fpa-utility-int}
    \end{equation}
    which has a unique form.

    Combining Eqns.~(\ref{eqn:fpa-utility}) and~(\ref{eqn:fpa-utility-int}), we conclude that, among all strictly increasing bidding strategies, the following is the unique candidate for a symmetric equilibrium:
    $$
    b(\theta)
    =
    \theta\nu_{a(\theta)}+t\mu_{a(\theta)}
    -
    \frac{
        c_{a(\theta)}-c_1
        +
        \displaystyle\int_0^\theta
        D(z)^{n-1}\nu_{a(z)}\,dz
    }{
        D(\theta)^{n-1}
    }
    $$
    for $\theta>0$.\footnote{\label{fn:fpa-b}Notice that the formula for $b(\theta)$ can also be derived from the second observation in the proof of Theorem~\ref{thm:revenue-equivalence}. Consider $\theta\in[0,\theta^\ast]$ for example. The expected payment in the second-price auction is $P(\theta)=\int_0^\theta (z\nu_1+t\mu_1)d(D(z)^{n-1})$ as derived in the proof of Theorem~\ref{thm:revenue-equivalence}. For the first price auction, the expected payment is $b(\theta)\cdot(D(\theta)^{n-1})$. The equation $b(\theta)\cdot(D(\theta)^{n-1})=\int_0^\theta (z\nu_1+t\mu_1)d(D(z)^{n-1})$ yields exactly (\ref{eqn:fpa-1}) by applying integration by parts. Equation (\ref{eqn:fpa-2}) can be obtained similarly.}
    For $\theta=0$, to make the bidding strategy continuous, we define 
    $$b(0):=\lim_{\theta\rightarrow0^+}b(\theta)=t\mu_1.$$

    We now prove claim 2 and verify that $b(\cdot)$ given above is continuous and strictly increasing. 
    First, consider an interval on which the action is constant, say $a(\theta)=j$. 
    On this interval,
    $$
        b(\theta)
        =
        \theta\nu_j+t\mu_j
        -
        \frac{
            c_j-c_1
            +
            \displaystyle\int_0^\theta D(z)^{n-1}\nu_{a(z)}\,dz
        }{
            D(\theta)^{n-1}
        }.
    $$
    It is straightforward to see its continuity.
    By taking the derivative, we obtain
    $$
    \begin{aligned}
        b'(\theta)
        &=
        \frac{
            (D(\theta)^{n-1})'
        }{
            D(\theta)^{2n-2}
        }
        \left(
            c_j-c_1
            +
            \int_0^\theta D(z)^{n-1}\nu_{a(z)}\,dz
        \right).
    \end{aligned}
    $$
    For every $\theta>0$, $c_j-c_1\ge 0$ and $\nu_{a(\theta)}>0$.
    Therefore, $b'(\theta)>0$ on every interval where the action is constant.

    It remains to show the continuity at $\theta^*$.
    The left limit of the bid at $\theta^\ast$ is
    $$
        b(\theta^{\ast-})
        =
        \theta^\ast\nu_1+t\mu_1
        -
        \frac{\displaystyle\int_0^{\theta^*}
        D(z)^{n-1}\nu_{a(z)}\,dz}{D(\theta^\ast)^{n-1}}.
    $$
    The right limit is
    $$
        b(\theta^{\ast+})
        =
        \theta^\ast\nu_2+t\mu_2
        -
        \frac{c_2-c_1+\displaystyle\int_0^{\theta^*}
        D(z)^{n-1}\nu_{a(z)}\,dz}{D(\theta^\ast)^{n-1}}.
    $$
    Thus, 
    $$
    \begin{aligned}
        b(\theta^{\ast+})
        -
        b(\theta^{\ast-})=
        \theta^\ast(\nu_2-\nu_1)+t(\mu_2-\mu_1)
        -
        \frac{c_2-c_1}{D(\theta^\ast)^{n-1}}=0,
    \end{aligned}
    $$
    where the last equality follows from Eqn.~(\ref{eqn:fpa-threshold}).
    This concludes the proof that the bidding strategy is continuous and strictly increasing in $\theta$.

    In the final step of claim $3$, we show that no profitable deviation exists. 
    Consider the case where a bidder with type $\theta$ deviates from the strategy stated above, while others follow it.
    When the bidder submits a bid $\tt b$, the expected utility for playing action $j$ is
    \begin{equation}
        \text{Probability of winning the auction}\times (\theta\nu_j+t\mu_j-{\tt b})-c_j,\label{eqn:fpa-deviation}
    \end{equation}
    where the probability of winning the auction depends only on $\tt b$ and is irrelevant to the chosen action.
    Therefore, instead of considering all possible pairs of actions and bids, it suffices to consider those pairs $(j,{\tt b})$ where $j$ is the action that maximizes Eqn.~(\ref{eqn:fpa-deviation}) under bid $\tt b$.

    Because $b(\cdot)$ is continuous and strictly increasing, every bid $\tt b$ in the interval $[b(0),b(B)]$ can be written uniquely as $b(z)$ for some $z\in[0,B]$. 
    If the bidder deviates to bid ${\tt b}=b(z)\in [b(0),b(B)]$, she will win with probability $D(z)^{n-1}$.
    Thus, her best expected utility from choosing the best action under bid $b(z)$ is
    $$
        \max_{j\in\{1,2\}}
        \left\{
            D(z)^{n-1}
            \left(
                \theta\nu_j+t\mu_j
            \right)
            -c_j
        \right\}
        -
        D(z)^{n-1}b(z).
    $$
    Define the value of the optimal action choice as
    $$
        \mathcal V(\theta,x)
        :=
        \max_{j\in\{1,2\}}
        \left\{
            x(\theta\nu_j+t\mu_j)-c_j
        \right\},
    $$
    where $\theta$ denotes the bidder's type and $x$ denotes the winning probability.   
    To prove the bidder does not deviate to bid $b(z)$, it suffices to prove that 
    $$\calU(\theta)\ge \mathcal V\left(\theta,D(z)^{n-1}\right)-D(z)^{n-1}b(z).$$
    Substituting with $\calU(z)=\mathcal V\left(z,D(z)^{n-1}\right)-D(z)^{n-1}b(z)$, it is equivalent to show
    \begin{equation}
        \calU(\theta)-\calU(z)\ge \mathcal V\left(\theta,D(z)^{n-1}\right)-\mathcal V\left(z,D(z)^{n-1}\right).\label{eqn:fpa-no-deviation}
    \end{equation}
    
    To prove Eqn.~(\ref{eqn:fpa-no-deviation}), the following lemma about $\mathcal V(\theta,x)$ is established.

    Let $a(\theta,x)$ denote the action that maximizes $\mathcal{V}(\theta,x)$, i.e.,
    $$
    a(\theta,x) \in \argmax_{j\in\{1,2\}}
    \left\{
    x(\theta\nu_j+t\mu_j)-c_j
    \right\}.
    $$

Equivalently, since there are only two actions, we can write it explicitly as
$$
a(\theta,x)=
\begin{cases}
1, & \text{if } x(\theta\nu_1+t\mu_1)-c_1 \ge x(\theta\nu_2+t\mu_2)-c_2,\\
2, & \text{otherwise}.
\end{cases}
$$
Note that $a(\theta,D(\theta)^{n-1}) = a(\theta)$.
    
    \begin{lemma}\label{lem:fpa-V}
    For every $\theta\in[0,B]$ where $\frac{\partial \mathcal V(\theta,x)}{\partial \theta}$ exists, we have
        $$
        \frac{\partial \mathcal V(\theta,x)}{\partial \theta}=x\nu_{a(\theta,x)},
        $$
        and $\frac{\partial \mathcal V(\theta,x)}{\partial \theta}$ is weakly increasing in $x$ almost everywhere for $x\in[0,1]$.
    \end{lemma}
    \begin{proof}
        For each fixed $x$, $\mathcal V(\cdot,x)$ is continuous and differentiable except possibly at one cutoff point. 
        According to the envelope theorem\footnote{The envelope theorem states that, for $f(x)=\max_y g(x,y)$ where $g(x,y)$ is continuously differentiable, it holds that $f'(x)=\frac{\partial g(x, y^*)}{\partial x}$ where $y^*$ is the maximizer to $g(x,y)$.}, it straightforwardly holds that
        $$
        \frac{\partial \mathcal V(\theta,x)}{\partial \theta}=x\nu_{a(\theta,x)}.
        $$
        Note that the difference in choosing the two actions
        $$(x(\theta\nu_2+t\mu_2)-c_2)-(x(\theta\nu_1+t\mu_1)-c_1)$$
        is weakly increasing in $x$ when $\theta$ is fixed.
        Therefore, increasing $x$ will only make action $2$ more attractive and lead to a possible change from action $1$ to action $2$ in $a(\theta,x)$.
        Since $\nu_2 > \nu_1$, the lemma concludes.
    \end{proof}

    We now turn to Eqn.~(\ref{eqn:fpa-no-deviation}), and first assume that $\theta\ge z$.
    Due to Eqn.~(\ref{eqn:fpa-utility-int}), we have
    $$\calU(\theta)-\calU(z)=\int_z^\theta D(s)^{n-1}\nu_{a(s)}\,ds=\int_z^\theta D(s)^{n-1}\nu_{a(s,D(s)^{n-1})}\,ds.$$
    Applying Lemma~\ref{lem:fpa-V}, we have
    \begin{align*}
        \calU(\theta)-\calU(z)&=
        \int_z^\theta \frac{\partial \mathcal V(s,D(s)^{n-1})}{\partial s}\,ds\\
        &\ge \int_z^\theta \frac{\partial \mathcal V(s,D(z)^{n-1})}{\partial s}\,ds\\
        &= \int_0^\theta \frac{\partial \mathcal V(s,D(z)^{n-1})}{\partial s}\,ds-\int_0^z \frac{\partial \mathcal V(s,D(z)^{n-1})}{\partial s}\,ds\\
        &= \mathcal V(\theta,D(z)^{n-1})-\mathcal V(z,D(z)^{n-1}),
    \end{align*}
    where the first equality holds due to Lemma~\ref{lem:fpa-V}, the next inequality holds due to the weakly increasing property for every $s\ge z$, and the last equality holds as $D(z)^{n-1}$ in the integral is a constant.

    Then we assume that $\theta<z$.
    By the same reasoning as in the previous case, we have
    \begin{align*}
        \calU(z)-\calU(\theta)&=
        \int_\theta^z D(s)^{n-1}\nu_{a(s)}\,ds\\
        &=\int_\theta^z \frac{\partial \mathcal V(s,D(s)^{n-1})}{\partial s}\,ds\\
        &\le \int_\theta^z \frac{\partial \mathcal V(s,D(z)^{n-1})}{\partial s}\,ds\\
        &= \mathcal V(z,D(z)^{n-1})-\mathcal V(\theta,D(z)^{n-1}).
    \end{align*}

    We have shown that no deviation to any bid ${\tt b}\in[b(0),b(B)]$ is profitable.
    It remains to consider ${\tt b}\notin[b(0),b(B)]$.
    If ${\tt b}<b(0)$, the bidder will lose the auction with probability $1$.
    The best expected utility is $-c_1$, which is the same as bidding $b(0)$, so she will not deviate.
    If ${\tt b}>b(B)$, the bidder will win the auction with probability $1$, which is strictly dominated by bidding $b(B)$, as the bidder will still win the auction with a strict decrease in payment from $\tt b$ to $b(B)$.
    This proves claim $3$.
\end{proof}

We remark here that the above equilibrium is unique among equilibria with strictly increasing bidding strategies, except that $b(0)$ may be chosen arbitrarily within $[0, t\mu_1]$ (as the bidder with type $0$ always loses the auction, any bid in $[0, t\mu_1]$ results in a utility of $-c_1$).

\section{Conclusion and Future Work}
\label{sec:conclusion}
We have proposed a new mechanism design framework by integrating a contract into an auction mechanism.
Our framework is a natural generalization of the auction theory when the quality of the transaction is concerned and the bidders face the risk of sunk cost for effort investment, and it is also a natural generalization of the principal-agent model with competitive environments of multiple agents.
We believe this is a natural economic framework that deserves further study.

We have focused on auction mechanisms that allocate the slot to the bidder with the highest bid, and worked out the optimal linear contracts under symmetric strictly increasing equilibria through a careful analysis of second-price auction and a generalized revenue equivalence theorem.
We demonstrate that the use of contracts can effectively mitigate the moral hazard issue and improve the revenue of the auctioneer.

There are multiple directions for future work.
A natural direction is to design the optimal mechanism together with the corresponding optimal linear contract that maximizes the auctioneer's expected revenue.
To reach an optimal revenue, the mechanism will violate the property that allocates the slot to the highest bid.
Recall that in the classical single-item auction setting without the contract, the Myerson auction achieves the optimal revenue by setting a reserve price for the bidders according to their type distributions.
We may also need to set a reserve price under the contract setting; however, our optimization seems more complicated.
The mechanism and the choice of $t$ under our model will jointly affect the allocation, the payment, as well as the bidders' incentives to choose actions.
Therefore, the optimal reserve price may further depend on the choice of $t$, instead of merely a function of type distributions.
We need to optimize the mechanism and $t$ jointly.

We have shown that under the second-price auction (and other auctions satisfying Theorem~\ref{thm:revenue-equivalence}), when both the quality reward $\Omega(\cdot)$ and contract $T(\cdot)$ have the linear form, it is optimal to set $t\approx\omega$, that is, it is optimal for the auctioneer to fully pass through her quality value to the winner.
It is natural to ask whether the full-pass-through conclusion continues to hold when $\Omega(\cdot)$ is nonlinear, when restricting the contract $T(\cdot)$ to have the same form as $\Omega(\cdot)$; if not, how does the optimal contract depend on the shape of $\Omega(\cdot)$

Other than this, another future direction is to explore the setting with richer contract forms or richer information settings.
On contracts, there is a rich literature that studies simple contracts (i.e., linear contracts) versus optimal contracts (see, e.g., \citet{dutting2019simple}).
It is known that simple linear contracts are optimal in many settings, and they are not in other settings.
It is interesting to see if linear contracts are optimal in our framework.
On information settings, bidders' types can be correlated, and bidders may have partial information, or even full information, about the others' types.
In this case, equilibrium analysis becomes more technically challenging, while our current technique for finding a Bayes Nash equilibrium in Theorem~\ref{thm:BNE-identical} relies on the symmetry of bidders' type distributions.

\bibliographystyle{plainnat}
\bibliography{reference}

\newpage
\appendix
\section{Omitted Details in Sect.~\ref{sec:spa}}
\subsection{Proof of Theorem~\ref{thm:BNE-existence}}
\label{sec:BNE-existence-proof}
Before proving the theorem, we first prove the following proposition, which shows that the best response must be a threshold strategy.
This proposition has the same spirit with Theorem~\ref{thm:Bayesian-EquilibriaCharacterization}, and it follows easily from Lemma~\ref{prop:monotone_diff}.

\begin{proposition}[Threshold best responses]
\label{prop:spa-threshold-best-response}
Consider the second-price auction with a monotone contract $T$ and any finite number $m$ of actions.
Fix a bidder $i$ and any natural-strategy profile $\Sigma^{(-i)}$ of the other $(n-1)$ bidders.
Let $U_j^{(i)}(\theta,\Sigma^{(-i)})$ denote bidder $i$'s interim expected utility from choosing action $j$ and bidding naturally when her type is $\theta$.
Then, every best-response action rule $a^{(i)}:\Theta^{(i)}\to[m]$, i.e., every action rule  $a^{(i)}$ satisfying 
\[
    \forall\theta\in\Theta^{(i)}:\quad a^{(i)}(\theta)
    \in
    \argmax_{j\in[m]}
    U_j^{(i)}(\theta,\Sigma^{(-i)})
\]
is weakly increasing in $\theta$.
\end{proposition}
\begin{proof}
By Lemma~\ref{prop:monotone_diff}, for every pair $r<s$, the difference
\[
    U_s^{(i)}(\theta,\Sigma^{(-i)})
    -
    U_r^{(i)}(\theta,\Sigma^{(-i)})
\]
is nondecreasing in $\theta$ and has at most one zero.

Suppose that a pointwise best-response action rule is not weakly increasing.
Then there exist $\theta'<\theta''$ and $r<s$ such that $a^{(i)}(\theta')=s$ and $a^{(i)}(\theta'')=r$.
Optimality at the two types and the preceding monotonicity give
\[
    0
    \le
    U_s^{(i)}(\theta',\Sigma^{(-i)})
    -
    U_r^{(i)}(\theta',\Sigma^{(-i)})
    \le
    U_s^{(i)}(\theta'',\Sigma^{(-i)})
    -
    U_r^{(i)}(\theta'',\Sigma^{(-i)})
    \le 0.
\]
Thus, the difference has two distinct zeros, a contradiction.
The action rule is therefore weakly increasing.
\end{proof}

Proposition~\ref{prop:spa-threshold-best-response} implies the best response is always a threshold strategy with at most $m-1$ thresholds, allowing skipped actions, constant-action strategies, and best-response choices at threshold types and support endpoints.
In other words, the best response is an increasing piecewise-constant function with at most $m-1$ points of discontinuity.
Moreover, $a^{(i)}(\cdot)$ is unique except at those points of discontinuity.

Now, we are ready to prove Theorem~\ref{thm:BNE-existence}.

\begin{proof}[Proof of Theorem~\ref{thm:BNE-existence}]
The case $m=1$ follows immediately from Proposition~\ref{prop:naturalbidding}, so assume $m\geq2$.

We search only over natural-strategy profiles whose action rules are threshold strategies.
Thus, we describe a strategy profile by only describing the action strategies that are piecewise-constant functions.
By Proposition~\ref{prop:spa-threshold-best-response}, when certifying a Bayes Nash equilibrium, it suffices to show that an action strategy is a best response \emph{among all threshold strategies}.

The high-level idea of the proof is as follows.
We will define a potential function $\Phi$ that maps an action profile $\mathbf{a}$ to a real value.
We make sure that the change in the expected utility of a bidder $i$ caused by a unilateral deviation is the same as the corresponding change in $\Phi$.
This will ensure the profile that maximizes $\Phi$ is a Bayes Nash equilibrium.
It then suffices to show that $\Phi$ has a maximum.
To show this, we will represent an action strategy $a^{(i)}:\Theta^{(i)}\to[m]$ by its $m-1$ thresholds.
Thus, the domain of $\Phi$ becomes a compact subset of $\mathbb{R}^{n(m-1)}$, and continuity of $\Phi$ (which will be proved) implies that $\Phi$ has a maximum.

\paragraph{A compact family of threshold strategies.}
We start by encoding an action strategy by an $(m-1)$-dimensional vector.
For each bidder $i$, define
\[
    \mathcal K_i
    :=
    \left\{
        (\tau_1^{(i)},\ldots,\tau_{m-1}^{(i)}):
        0\leq\tau_1^{(i)}\leq\cdots\leq\tau_{m-1}^{(i)}
        \leq B^{(i)}
    \right\},
    \qquad
    \mathcal K:=\prod_{i=1}^n\mathcal K_i.
\]
For $\tau\in\mathcal K$, let bidder $i$ use the action rule
\[
    a_\tau^{(i)}(\theta)
    :=
    1+\sum_{j=1}^{m-1}
    \mathbf{1}\{\theta>\tau_j^{(i)}\},
\]
and let $\mathbf a_\tau=(a_\tau^{(1)},\ldots,a_\tau^{(n)})$.
The corresponding natural bids are understood throughout.

Repeated thresholds allow actions to be skipped.
Thresholds at zero or $B^{(i)}$ allow constant-action strategies, up to their values at the endpoints.
Thus, every weakly increasing action rule is represented by $\mathcal K_i$ except possibly at finitely many threshold types and support endpoints.
Because the type distributions are atomless, these exceptional values do not affect expectations.
We will choose them appropriately at the end of the proof.

\paragraph{The potential function.}
For a threshold-strategy profile $\mathbf a=(a^{(1)},\ldots,a^{(n)})$ and a type profile
$\boldsymbol\theta=(\theta^{(1)},\ldots,\theta^{(n)})$, define
\[
    W_i
    :=
    \theta^{(i)}\cdot\nu_{a^{(i)}(\theta^{(i)})}+\beta_{a^{(i)}(\theta^{(i)})},
    \qquad
    M_i:=\max_{i'\neq i}W_{i'}.
\]
After taking the expectation over outcome quality, bidder $i$'s utility at this type profile is
\[
    u^{(i)}(\boldsymbol\theta;\mathbf a)
    =
    (W_i-M_i)_+-c_{a^{(i)}(\theta^{(i)})},
\]
where $x_+:=\max\{x,0\}$.
This formula also holds when bids tie, because a bidder whose natural bid equals the highest competing bid has zero surplus before subtracting her action cost.

Define
\[
    \Phi(\mathbf a)
    :=
    \mathbb{E}_{\boldsymbol\theta}
    \left[
        \max_{i\in[n]}W_i
        -
        \sum_{i=1}^n c_{a^{(i)}(\theta^{(i)})}
    \right].
\]
We claim that any unilateral change in a bidder's action rule changes $\Phi$ by exactly the change in her ex-ante expected utility.

Indeed, fix bidder $i$, the opponents' action rules, and a type profile.
The highest competing bid $M_i$ remains unchanged.
Using
\[
    \max\{W_i,M_i\}=M_i+(W_i-M_i)_+,
\]
the change in the integrand defining $\Phi$ is exactly the change in bidder $i$'s utility.
Therefore, for any alternative threshold action rule
$\widetilde a^{(i)}$,
\begin{equation}\label{eq:spa-general-m-potential}
    \Phi(\widetilde a^{(i)},\mathbf a^{(-i)})
    -\Phi(\mathbf a)
    =
    \mathbb{E}_{\theta^{(i)}}
    \left[
        U_{\widetilde a^{(i)}(\theta^{(i)})}^{(i)}
        (\theta^{(i)};\mathbf a^{(-i)})
        -
        U_{a^{(i)}(\theta^{(i)})}^{(i)}
        (\theta^{(i)};\mathbf a^{(-i)})
    \right],
\end{equation}
where $U_j^{(i)}(\theta;\mathbf a^{(-i)})$ denotes the interim utility from action $j$ and its natural bid against the opponents' natural strategies induced by $\mathbf a^{(-i)}$.

\paragraph{The potential attains a maximum.}
The set $\mathcal K$ is compact.
We show that $\tau\mapsto\Phi(\mathbf a_\tau)$ is continuous.

Consider a sequence $\tau^{[\ell]}\to\tau$ in $\mathcal K$.
For every bidder $i$ and every type $\theta$ different from the finitely many limiting thresholds $\tau_1^{(i)},\ldots,\tau_{m-1}^{(i)}$, we have
\[
    a_{\tau^{[\ell]}}^{(i)}(\theta)
    =
    a_\tau^{(i)}(\theta)
    \qquad\text{for all sufficiently large }\ell.
\]
Thus, the integrand defining $\Phi(\mathbf a_{\tau^{[\ell]}})$ converges almost surely to the integrand defining $\Phi(\mathbf a_\tau)$.
The integrands are uniformly bounded, since the type supports are bounded, the action set is finite, and all $\nu_j,\beta_j,c_j$ are finite.
Dominated convergence implies
\[
    \Phi(\mathbf a_{\tau^{[\ell]}})
    \longrightarrow
    \Phi(\mathbf a_\tau).
\]
Hence, $\Phi(\mathbf a_\tau)$ attains a maximum on $\mathcal K$.
Let $\tau^\ast$ be a maximizer and write $\mathbf a^\ast:=\mathbf a_{\tau^\ast}$.

\paragraph{Optimality almost everywhere.}
Fix bidder $i$ and choose the pointwise best response
\[
    \widehat a^{(i)}(\theta)
    :=
    \min\argmax_{j\in[m]}
    U_j^{(i)}(\theta;\mathbf a^{\ast(-i)}).
\]
By Proposition~\ref{prop:spa-threshold-best-response}, $\widehat a^{(i)}$ is a threshold strategy.
Up to its values at finitely many types, it is represented in $\mathcal K_i$.
Maximality of $\tau^\ast$ and \eqref{eq:spa-general-m-potential} therefore give
\[
\begin{aligned}
    0
    &\geq
    \Phi(\widehat a^{(i)},\mathbf a^{\ast(-i)})
    -\Phi(\mathbf a^\ast)\\
    &=
    \mathbb{E}_{\theta^{(i)}}
    \left[
        \max_{j\in[m]}
        U_j^{(i)}(\theta^{(i)};\mathbf a^{\ast(-i)})
        -
        U_{a^{\ast(i)}(\theta^{(i)})}^{(i)}
        (\theta^{(i)};\mathbf a^{\ast(-i)})
    \right]\\
    &\geq0.
\end{aligned}
\]
The nonnegative quantity inside the expectation must therefore equal zero almost surely.
Consequently, for every bidder $i$,
\[
    a^{\ast(i)}(\theta)
    \in
    \argmax_{j\in[m]}
    U_j^{(i)}(\theta;\mathbf a^{\ast(-i)})
    \qquad\text{for almost every }\theta.
\]
For now, we have already proved the theorem if the Bayes Nash equilibria are defined in ex-ante.
The last part below proves the existence of the interim version of the equilibrium defined in this paper.

\paragraph{Optimality at every type.}
Replace each $a^{\ast(i)}$ by the pointwise best response $\widehat a^{(i)}$ defined above.
By Proposition~\ref{prop:spa-threshold-best-response}, the maximizing action against $\mathbf a^{\ast(-i)}$ is unique except at finitely many types.
Together with the preceding almost-sure optimality, this implies
\[
    \widehat a^{(i)}(\theta)=a^{\ast(i)}(\theta)
    \qquad\text{for almost every }\theta.
\]
Thus, the replacements leave every bidder's distribution of natural bids unchanged.

Let $\widehat{\mathbf a}$ be the resulting profile.
By independence of types, every bidder faces the same distribution of opponents' bids under $\widehat{\mathbf a}^{(-i)}$ as under $\mathbf a^{\ast(-i)}$, for every possible own type.
Hence,
\[
    U_j^{(i)}(\theta;\widehat{\mathbf a}^{(-i)})
    =
    U_j^{(i)}(\theta;\mathbf a^{\ast(-i)})
\]
for every bidder $i$, type $\theta$, and action $j$.
It follows that $\widehat a^{(i)}(\theta)$ is a best response against $\widehat{\mathbf a}^{(-i)}$ at every type.
\end{proof}

\subsection{Non-uniqueness of Asymmetric Bayes Nash Equilibria}
\label{sec:non-uniqueness}
In this section, we consider the second-price auction and present an example where bidders' type distributions are the uniform distribution on $[0,1]$ and there are two equilibria: besides the symmetric one captured by Theorem~\ref{thm:BNE-identical}, there is also an asymmetric equilibrium.
The intuitions for the existence of asymmetric equilibria are as follows.

Consider two bidders and two actions.
Consider the strategy profile with thresholds $(\theta^{\ast(1)},\theta^{\ast(2)})$,
where bidder $i$ plays a threshold strategy with threshold $\theta^{\ast(i)}$. 
By Lemma~\ref{prop:monotone_diff}, we must have $U_1^{(1)}(\theta^{\ast(1)})=U_2^{(1)}(\theta^{\ast(1)})$ and $U_1^{(2)}(\theta^{\ast(2)})=U_2^{(2)}(\theta^{\ast(2)})$.
That is, at bidder $i$'s threshold $\theta_i$, bidder $i$ has the same utility for playing action $1$ and action $2$, supposing bidder $3-i$ has a type sampled from the uniform distribution on $[0,1]$ and plays a fixed threshold strategy with $\theta_{3-i}$.
We have two equations for two variables $\theta^{\ast(1)},\theta^{\ast(2)}$, and we know one solution with $\theta^{\ast(1)}=\theta^{\ast(2)}$ captured by Theorem~\ref{thm:BNE-identical}.
However, there may be another solution, as each of the two equations contains quadratic terms (for $n=2$).
We can carefully construct an instance where the other solution is also valid.

Now we describe a concrete example.
The number of bidders is $n=2$, and the number of actions is $m=2$.
Suppose playing action $1$ yields quality $9$ with probability $1$ and playing action $2$ yields quality $11$ with probability $1$.
The costs for actions $1$ and $2$ are $0.5$ and $2$ respectively.
Bidders' valuation is $v(\theta,q)=\theta\cdot \vq(q)$ with $\vq(q)=q$.
A linear contract is used with parameter $t=1$.
In this example, $\nu_1=\mu_1=9$, $\nu_2=\mu_2=11$, $c_1=0.5$, $c_2=2$, and $t=1$.

We show that bidder $1$ playing the threshold strategy with $\theta^{\ast(1)}=\frac35$ and bidder $2$ playing the threshold strategy with $\theta^{\ast(2)}=\frac25$ form a Bayes Nash equilibrium.
By Lemma~\ref{prop:monotone_diff}, it suffices to show that $U_1^{(1)}(\theta^{\ast(1)})=U_2^{(1)}(\theta^{\ast(1)})$ and $U_1^{(2)}(\theta^{\ast(2)})=U_2^{(2)}(\theta^{\ast(2)})$.
We verify these two equations below.

We first consider bidder $1$ and assume her type is exactly $\theta^{\ast(1)}=\frac35$.
When she plays action $1$, she wins if and only if bidder $2$'s type is below $\theta^{\ast(2)}=\frac25$.
To see this, if bidder $2$'s type is above $\theta^{\ast(2)}=\frac25$, bidder $2$ will play action $2$ and bids at least $\theta^{\ast(2)}\nu_2+t\mu_2=15.4$, which is more than bidder $1$'s bid $\theta^{\ast(1)}\nu_1+t\mu_1=14.4$.
If bidder $2$'s type is below $\theta^{\ast(2)}=\frac25$, bidder $2$ will play action $1$, and bidder $1$ wins since $\theta^{\ast(1)}>\theta^{\ast(2)}$.
Therefore, bidder $1$ wins with probability $\frac25$ given the uniform distribution.
The expected utility for bidder $1$ when playing action $1$ is
$$-c_1+\frac25\left(\theta^{\ast(1)}\nu_1+t\mu_1\right)-\int_0^{\frac25}(z\nu_1+t\mu_1)dz,$$
where the second term is the expected value and reward, and the third term is the expected payment.
Plugging in the values of the variables, the expected utility is
$$\frac{47}{50}.$$

When bidder $1$ plays action $2$, she wins if and only if bidder $2$'s type is below $\theta^{\ast(1)}=\frac35$.
The expected value and reward, combined, are
$$\frac35(\theta^{\ast(1)}\nu_2+t\mu_2)=\frac{264}{25}.$$
The expected payment is
$$\int_0^{\theta^{\ast(2)}}(z\nu_1+t\mu_1)dz+\int_{\theta^{\ast(2)}}^{\theta^{\ast(1)}}(z\nu_2+t\mu_2)dz=\frac{381}{50}.$$
The expected utility is
$$-c_2+\frac{264}{25}-\frac{381}{50}=\frac{47}{50},$$
which agrees with the utility for playing action $1$.

Now we analyze bidder $2$ and assume her type is exactly $\theta^{\ast(2)}=\frac25$.
When she plays action $1$, she wins if and only if bidder $1$'s type is below $\frac25$.
The expected value and reward, combined, are
$$\frac25\left(\theta^{\ast(2)}\nu_1+t\mu_1\right)=\frac{126}{25}.$$
The expected payment is
$$\int_{0}^{\theta^{\ast(2)}}(z\nu_1+t\mu_1)dz=\frac{108}{25}.$$
The expected utility is
$$-c_1+\frac{126}{25}-\frac{108}{25}=\frac{11}{50}.$$

When bidder $2$ plays action $2$, she wins if and only if bidder $1$'s type is below $\theta^{\ast(1)}$ for similar reasons as stated earlier.
The expected value and reward, combined, are
$$\frac35\left(\theta^{\ast(2)}\nu_2+t\mu_2\right)=\frac{231}{25}.$$
The expected payment is
$$\int_{0}^{\theta^{\ast(1)}}(z\nu_1+t\mu_1)dz=\frac{351}{50}.$$
The expected utility is
$$-c_2+\frac{231}{25}-\frac{351}{50}=\frac{11}{50},$$
which agrees with the utility when action $1$ was played.

\subsection{Optimal Contract for Large Finite $n$ with Asymptotic Analysis}
\label{sec:BNE-optimalt-finite-n}
In this section, we investigate the optimal contract with finite $n$ with asymptotic analysis.
We consider the special case where the distribution is a uniform distribution on $[0,1]$.
This is for the ease of computation.
One of our main purposes is to show that $t$ is not exactly $\omega$ in the second case.
Eqn.~(\ref{eqn:uniform-firstorder}) implies $\theta^\ast\to1$ as $n\to\infty$, so $t\to\omega$.
However, Eqn.~(\ref{eqn:uniform-firstorder}) also shows that $\theta^\ast$ is not exactly $1$ for finite $n$ if $\omega>\frac{(c_2-c_1)-(\nu_2-\nu_1)}{\mu_2-\mu_1}$.

\begin{theorem}\label{thm:BNE-optimalt-uniform}
    Consider the second-price auction.
    When bidders' type distributions are identically uniform distribution on $[0,1]$, and there are two actions $1$ and $2$, consider linear contract $T(q)=tq$ where $t$ maximizes the expected revenue subject to that bidders play the unique symmetric Bayes Nash equilibrium.
    We have
    \begin{itemize}
        \item If $\omega\le\frac{(c_2-c_1)-(\nu_2-\nu_1)}{\mu_2-\mu_1}$, then $t$ can be any value in $\left[0,\frac{(c_2-c_1)-(\nu_2-\nu_1)}{\mu_2-\mu_1}\right]$ for any $n\geq2$. 
        \item If $\omega>\frac{(c_2-c_1)-(\nu_2-\nu_1)}{\mu_2-\mu_1}$, for $n\geq\log_2\left(\frac{2(\nu_2-\nu_1)}{c_2-c_1}\right)$, we have
        $$t=\max\left\{\omega-\frac{(\nu_2-\nu_1)(1-\theta^\ast)}{\mu_2-\mu_1},0\right\},$$
        where $\theta^\ast$ is the solution to the equation
        \begin{equation}\label{eqn:uniform-firstorder}
    2(\nu_2-\nu_1)(\theta^\ast)^n+(\omega(\mu_2-\mu_1)-(\nu_2-\nu_1))(\theta^\ast)^{n-1}=c_2-c_1.
        \end{equation}
    \end{itemize}
\end{theorem}
\begin{proof}

We compute the expected revenue similar to the proof of Theorem~\ref{thm:BNE-optimalt}, but without substituting $D(\theta^\ast)=1-\frac xn$ and $(D(\theta^\ast))^n=e^{-x}$ (which only holds for the limit case $n\to\infty$).
With the uniform distribution assumption, we can explicitly compute the two integrals in the first term $R_1$ of the revenue.

For uniform distribution, by taking $D(\theta^\ast)=\theta^\ast$ in Eqn.~(\ref{eqn:thetastar}), we have
\begin{equation}\label{eqn:thetastar-uniform}
    (\nu_2-\nu_1)(\theta^\ast)^{n}+t(\mu_2-\mu_1)(\theta^\ast)^{n-1}=c_2-c_1.
\end{equation}

Let $\pi$ be the probability density function of the distribution of the second maximum number of $n$ independent random variables with distribution Unif$(0,1)$.
Then
$$ \pi(z) = n(n-1) \cdot (1-z) \cdot z^{n-2}.$$
The expectation for the first term of the revenue is
$$R_1(\theta^\ast)=\int_0^{\theta^\ast}\pi(z)z\nu_1 dz+\int_{\theta^\ast}^1\pi(z)z\nu_2 dz=\frac{n-1}{n+1}\nu_2-(\nu_2-\nu_1)(n-1)(\theta^\ast)^{n}+(\nu_2-\nu_1)\frac{n(n-1)}{n+1}(\theta^\ast)^{n+1}.$$
The sum of the second and the third terms, by taking $D(\theta^\ast)=\theta^\ast$ for uniform distribution in (\ref{eqn:gamma}), is 
$$R_{23}(\theta^\ast)=-n(1-\theta^\ast)\left(c_2-c_1-(\theta^\ast)^{n}(\nu_2-\nu_1)\right).$$
The last term is
$$R_4(\theta^\ast)=(\theta^\ast)^n\cdot\omega\mu_1+(1-(\theta^\ast)^n)\cdot\omega\mu_2=\omega\mu_2-\omega(\mu_2-\mu_1)\cdot(\theta^\ast)^n.$$
Combining together, the expected revenue is
$$R(\theta^\ast)=\frac{n-1}{n+1}\nu_2+\omega\mu_2-n(c_2-c_1)+n(c_2-c_1)\theta^\ast+((\nu_2-\nu_1)-\omega(\mu_2-\mu_1))(\theta^\ast)^n-\frac{2n}{n+1}(\nu_2-\nu_1)(\theta^\ast)^{n+1}.$$
Denote the above by $\varphi(\theta^\ast)$.
We have
$$\varphi'(\theta^\ast)=n(c_2-c_1)+n((\nu_2-\nu_1)-\omega(\mu_2-\mu_1))(\theta^\ast)^{n-1}-2n(\nu_2-\nu_1)(\theta^\ast)^n,$$
and
$$\varphi''(\theta^\ast)=n(n-1)((\nu_2-\nu_1)-\omega(\mu_2-\mu_1))(\theta^\ast)^{n-2}-2n^2(\nu_2-\nu_1)(\theta^\ast)^{n-1}.$$

For the first case with $\omega\le\frac{(c_2-c_1)-(\nu_2-\nu_1)}{\mu_2-\mu_1}$ and for $\theta^\ast\in(0,1]$, we have
\begin{align*}
    \varphi'(\theta^\ast)&\ge n\left(c_2-c_1+\left(\nu_2-\nu_1-(\mu_2-\mu_1)\frac{(c_2-c_1)-(\nu_2-\nu_1)}{\mu_2-\mu_1}\right)(\theta^\ast)^{n-1}-2(\nu_2-\nu_1)(\theta^\ast)^n\right)\\
    &=n\left(c_2-c_1+(-(c_2-c_1)+2(\nu_2-\nu_1))(\theta^\ast)^{n-1}-2(\nu_2-\nu_1)(\theta^\ast)^n\right)\\
    &\geq n\left(c_2-c_1+(-(c_2-c_1)+2(\nu_2-\nu_1))(\theta^\ast)^{n-1}-2(\nu_2-\nu_1)(\theta^\ast)^{n-1}\right)\tag{since $\theta^\ast\leq1$}\\
    &=n\left((c_2-c_1)(1-(\theta^\ast)^{n-1})\right)\\
    &\geq0.
\end{align*}
This implies $\varphi(\theta^\ast)$ is optimized at the right boundary $\theta^\ast=1$.
In this case, all bidders always play action $1$.
The expected revenue is
$$R(1)=\frac{n-1}{n+1}\nu_1+\omega\mu_1,$$
and this holds for all $t$ such that $\theta^\ast=1$.
By Eqn.~(\ref{eqn:thetastar-uniform}), we can set $t$ to be any value in $\left[0,\frac{(c_2-c_1)-(\nu_2-\nu_1)}{\mu_2-\mu_1}\right]$.

For the second case $\omega>\frac{(c_2-c_1)-(\nu_2-\nu_1)}{\mu_2-\mu_1}$, we have
$$\varphi'(0)=n(c_2-c_1)>0$$
and
$$\varphi'(1)=n((c_2-c_1)-(\nu_2-\nu_1)-\omega(\mu_2-\mu_1))\leq0.$$
The first-order condition $\varphi'(\theta^\ast)=0$ has a solution $\theta^\ast$ in $[0,1]$, and $\varphi'(\theta^\ast)=0$ yields Eqn.~(\ref{eqn:uniform-firstorder}).

On the other hand, by noticing $n^2>n(n-1)$, we have
\begin{equation}
    \varphi''(\theta^\ast)<(\theta^\ast)^{n-2}n(n-1)\cdot\left(-\omega(\mu_2-\mu_1)-(2\theta^\ast-1)(\nu_2-\nu_1)\right),
\end{equation}
which is negative for $\theta^\ast>0.5$.
Next, we will show that, for sufficiently large $n$, the solution to Eqn.~(\ref{eqn:uniform-firstorder}) is more than $0.5$, which gives the optimal revenue.

Firstly, if $\omega(\mu_2-\mu_1)\geq(\nu_2-\nu_1)$, we have
$$\varphi''(\theta^\ast)<(\theta^\ast)^{n-2}n(n-1)\cdot(-2\theta^\ast(\nu_2-\nu_1))<0$$
for $\theta^\ast\in(0,1]$.
The solution to Eqn.~(\ref{eqn:uniform-firstorder}) optimizes the revenue.
If $\omega(\mu_2-\mu_1)<(\nu_2-\nu_1)$, the coefficient $(\omega(\mu_2-\mu_1)-(\nu_2-\nu_1))$ of the second term on the left-hand side of Eqn.~(\ref{eqn:uniform-firstorder}) is negative, so the solution of (\ref{eqn:uniform-firstorder}) is larger than the solution to
$$2(\nu_2-\nu_1)(\theta^\ast)^n=c_2-c_1.$$
Therefore, to ensure that the solution of (\ref{eqn:uniform-firstorder}) is larger than $0.5$, it suffices to set $n$ such that
$$2(\nu_2-\nu_1)0.5^n\leq c_2-c_1,$$
which is
$$n\geq\log_2\left(\frac{2(\nu_2-\nu_1)}{c_2-c_1}\right).$$
Therefore, if $n$ satisfies the inequality above, the solution of (\ref{eqn:uniform-firstorder}) maximizes $\varphi(\cdot)$.

On the other hand, by taking $D(\theta^\ast)=\theta^\ast$ in Eqn.~(\ref{eqn:thetastar}), we have
$$(\nu_2-\nu_1)(\theta^\ast)^{n}+t(\mu_2-\mu_1)(\theta^\ast)^{n-1}=c_2-c_1$$
Combining this with Eqn.~(\ref{eqn:uniform-firstorder}), we have
$$(\nu_2-\nu_1)(\theta^\ast)^{n}+((\omega-t)(\mu_2-\mu_1)-(\nu_2-\nu_1))(\theta^\ast)^{n-1}=0,$$
which implies
$$(\nu_2-\nu_1)(1-\theta^\ast)=(\omega-t)(\mu_2-\mu_1),$$
and
$$t=\omega-\frac{(\nu_2-\nu_1)(1-\theta^\ast)}{\mu_2-\mu_1}.$$
The theorem follows by ensuring $t$ above is not negative.
\end{proof}

\section{Second-Price Auction: Extensions to $m$ Actions}
\label{sect:Incomplete-m-actions}
In this section, we extend our results in Sect.~\ref{sec:spa} to the more general case with $m$ actions.
As the main take-away message, the optimal reward parameter $t$ is still $\omega$ for more than two actions (again assuming $n\to\infty$).

We first extend our characterization of Bayes Nash equilibria (Theorem~\ref{thm:Bayesian-EquilibriaCharacterization}).
Before this, we describe an action-threshold-selection algorithm that will be used in later theorems.

Consider an arbitrary bidder $i$ and fix the strategies $\Sigma^{(-i)}$ of the remaining bidders.
Below, we omit the superscript $(i)$.
Upon $\Sigma^{(-i)}$, for any two actions $j_1,j_2$ with $j_1<j_2$, Lemma~\ref{prop:monotone_diff} implies there is a threshold $\theta_{j_1,j_2}$ such that
$$U_{j_1}(\theta,\Sigma^{(-i)})>U_{j_2}(\theta,\Sigma^{(-i)})\qquad\mbox{if }\theta<\theta_{j_1,j_2}$$
and
$$U_{j_1}(\theta,\Sigma^{(-i)})<U_{j_2}(\theta,\Sigma^{(-i)})\qquad\mbox{if }\theta>\theta_{j_1,j_2}.$$
We then have a collection of $\binom{m}{2}$ thresholds $\{\theta_{j_1,j_2}:j_1<j_2\}$.
This collection will yield a threshold strategy for bidder $i$ where a subset of actions $A=\{j_1,j_2,\ldots,j_{m'}\}\subseteq[m]$ (with $j_1<j_2<\cdots<j_{m'}$) is available such that actions with larger indices are played for larger types of bidder $i$.

Let us consider the examples with $m=3$.
We have three thresholds $\theta_{1,2},\theta_{1,3},\theta_{2,3}$.
If $0<\theta_{1,2}<\theta_{1,3}<\theta_{2,3}<B$, then action $1$ is the best response for $\theta<\theta_{1,2}$.
For $\theta\in(\theta_{1,2},\theta_{1,3})$, action $2$ is better than action $1$ since $\theta>\theta_{1,2}$.
On the other hand, action $2$ is also better than action $3$ since $\theta<\theta_{1,3}<\theta_{2,3}$.
Therefore, action $2$ is the best response.
For $\theta\in(\theta_{1,3},\theta_{2,3})$, similar analysis gives that action $2$ continues to be the best response.
For $\theta\in(\theta_{2,3},B]$, action $3$ is the best response.
Therefore, the two thresholds separating the three actions are $\theta_{1,2}$ and $\theta_{2,3}$.
The threshold $\theta_{1,3}$ is redundant: this is the point where playing action $3$ is as good as playing action $1$, but playing action $2$ is better.

For another example, consider $0<\theta_{2,3}<\theta_{1,3}<\theta_{1,2}<B$.
Before we reach $\theta=\theta_{2,3}$, the thresholds $\theta_{1,3}$ and $\theta_{1,2}$ are still not met.
Playing action $1$ is still the best response even when the utility for playing action $3$ begins to meet the utility for playing action $2$.
Thus, this threshold $\theta_{2,3}$ is redundant.
Moreover, we must have $\theta_{1,3}<\theta_{1,2}$ if $\theta_{2,3}$ is smallest among the three threshold.
To see this, we always have $U_3(\theta,\Sigma^{(-i)})>U_2(\theta,\Sigma^{-i})$ for $\theta>\theta_{2,3}$ by Lemma~\ref{prop:monotone_diff}; as a result, before we reach $\theta_{1,3}$ where $U_3(\theta,\Sigma^{(-i)})=U_1(\theta,\Sigma^{(-i)})$, action $2$ is still ``left behind'', with $U_2(\theta,\Sigma^{-i})<U_3(\theta,\Sigma^{-i})=U_1(\theta,\Sigma^{-i})$.
Therefore, for $\theta<\theta_{1,3}$, playing action $1$ is the best response; for $\theta>\theta_{1,3}$, playing action $3$ is the best response.
In this case, action $2$ is never played, and the two thresholds $\theta_{2,3}$ and $\theta_{1,2}$ are redundant.

Generalizing this idea, we have Algorithm~\ref{alg:selection} which, given a collection of $\binom{m}{2}$ thresholds, outputs a set $A$ of non-redundant actions and the corresponding thresholds separating them.
Notice that the algorithm only outputs the available actions: $j_1,\ldots,j_{m'}$ with $j_1<j_2<\cdots<j_{m'}$, and it is understood that the non-redundant thresholds separating them are $\theta_{j_1,j_2},\theta_{j_2,j_3},\ldots,\theta_{j_{m'-1},j_{m'}}$.
Below, we assume $\theta_{j_1,j_2}$ is defined for every pair $(j_1,j_2)$ with $j_1<j_2$.
If $U_{j_1}(\theta,\Sigma^{(-i)})<U_{j_2}(\theta,\Sigma^{(-i)})$ for all $\theta\in[0,B]$, we set $\theta_{j_1,j_2}=-1$ by convention.
Similarly, If $U_{j_1}(\theta,\Sigma^{(-i)})>U_{j_2}(\theta,\Sigma^{(-i)})$ for all $\theta\in[0,B]$, we set $\theta_{j_1,j_2}=B+1$.

\begin{algorithm}[H]
\caption{Action-Threshold-Selection Algorithm}
\label{alg:selection}
\begin{algorithmic}[1]
\REQUIRE a set of $\binom{m}{2}$ thresholds $\{\theta_{j_1,j_2}:j_1,j_2\in[m], j_1<j_2\}$
\ENSURE an ordered set $A$ of non-redundant actions \\
\selection$(\{\theta_{j_1,j_2}:j_1,j_2\in[m], j_1<j_2\})$:
\STATE initialize $X=[m]$ and $A=\emptyset$
\WHILE{$X\neq\emptyset$}
    \STATE let $j$ be the action in $X$ with the smallest index
    \IF{$j$ is the unique action in $X$}
        \STATE add $j$ to $A$
        \STATE \textbf{break}
    \ENDIF
    \STATE let $j^\ast$ be the action in $X$ with the smallest $\theta_{j,j^\ast}$
    \STATE remove all $j'$ with $j'<j^\ast$ from $X$
    \IF{$\theta_{j,j^\ast}>0$}
        \STATE add $j$ to $A$
    \ENDIF
    \IF{$\theta_{j,j^\ast}>B$}
        \STATE remove all actions from $X$
    \ENDIF
\ENDWHILE
\STATE \textbf{return} $A$
\end{algorithmic}
\end{algorithm}

Algorithm~\ref{alg:selection} does the following.
Each time, we consider an available action $j$ with the smallest index.
We find the action $j^\ast$ such that, for the first time, action $j$ is surpassed by action $j^\ast$ as $\theta$ increases.
This is the action $j^\ast$ with the smallest $\theta_{j,j^\ast}$.
Notice that all actions between $j$ and $j^\ast$ are redundant: if $j^\ast$ begins to dominate $j$ prior to any actions between $j$ and $j^\ast$, $j^\ast$ is better than all these actions later on, and these actions can never be the best responses.
If $\theta_{j,j^\ast}<0$, this means $j^\ast$ is always better than $j$ in the domain $[0,B]$, and $j$ is redundant.
Otherwise, we add $j$ as a non-redundant action, and the next iteration of the algorithm will check $j^\ast$ by definition (and see which action dominates $j^\ast$ first).
If $\theta_{j,j^\ast}>B$, this means $j$ remains the best response later on and all the way to the end.
We have located all the non-redundant actions and the algorithm will terminate and output $A$ after the inclusion of $j$.

\begin{theorem}[Characterization for BNE]\label{thm:Bayesian-EquilibriaCharacterization-m}
    Consider the second-price auction with monotone $T(\cdot)$.
    In any Bayes Nash equilibrium, each bidder $i$'s strategy is characterized by a set of actions $A^{(i)}=\{j_1^{(i)},\ldots,j_{m^{'(i)}}^{(i)}\}$ with $j_1^{(i)}<j_2^{(i)}<\cdots<j_{m^{'(i)}}^{(i)}$ and a corresponding set of thresholds $\theta_{j_1,j_2}^{(i)},\theta_{j_2,j_3}^{(i)},\ldots,\theta_{j_{m^{'(i)}-1},j_{m^{'(i)}}}^{(i)}$ such that action $j_k$ is played if $\theta^{(i)}\in[\theta_{j_{k-1},j_k}^{(i)},\theta_{j_{k},j_{k+1}}^{(i)})$, where we set $\theta_{j_0,j_1}=0$ and $\theta_{j_{m^{'(i)}}, j_{m^{'(i)}+1}}=B^{(i)}+1$.
\end{theorem}
\begin{proof}
    It can be proved based on Lemma~\ref{prop:monotone_diff}.
    Consider a Bayes Nash equilibrium $\Sigma$. Fixed $\Sigma^{(-i)}$.
    We show that the best response to $\Sigma^{(-i)}$ must be a threshold strategy.
    Below, we omit the superscript $(i)$.
    For every pair of actions $j_1,j_2$ with $j_1<j_2$, define $\theta_{j_1,j_2}$ be the $\theta$ such that $U_{j_1}(\theta,\Sigma^{(-i)})=U_{j_2}(\theta,\Sigma^{(-i)})$.
    If $U_{j_1}(\theta,\Sigma^{(-i)})<U_{j_2}(\theta,\Sigma^{(-i)})$ for all $\theta\in[0,B]$, we set $\theta_{j_1,j_2}=-1$.
    Similarly, If $U_{j_1}(\theta,\Sigma^{(-i)})>U_{j_2}(\theta,\Sigma^{(-i)})$ for all $\theta\in[0,B]$, we set $\theta_{j_1,j_2}=B+1$.
    
    As $\theta$ increases, by Lemma~\ref{prop:monotone_diff}, $U_j(\theta)$ increases faster for larger $j$.
    By the time $U_j(\theta)$ becomes the maximum among $U_1(\theta),\ldots,U_m(\theta)$, action $j$ is optimal at this moment, and actions $1,\ldots,j-1$ will always be suboptimal as $\theta$ further increases (since $U_j(\theta)-U_{j'}(\theta)$ is strictly increasing for any $j'<j$).
    Following our discussions prior to this theorem, the best response to $\Sigma^{(-i)}$ is a threshold strategy by applying \selection$(\{\theta_{j_1,j_2}\})$ in Algorithm~\ref{alg:selection}.
\end{proof}

Next, we consider identical type distribution and extend Theorem~\ref{thm:BNE-identical} to the case with general $m$.

\begin{theorem}[Unique Symmetric BNE for Identical Type Distributions]\label{thm:BNE-identical-m}
    Consider the second-price auction with a linear contract $T(q)=tq$.
    When bidders' type distributions are identical, there is a symmetric Bayes Nash equilibrium where each bidder's strategy is a threshold strategy with available action sets $A=\{j_1,\ldots,j_{m'}\}$ (where $j_1<j_2<\cdots<j_{m'}$) and thresholds $\theta_{j_1,j_2},\theta_{j_2,j_3},\ldots,\theta_{j_{m'}-1,j_{m'}}$, where $A$ and the thresholds are given by 
    $$\text{\selection}(\{\theta_{j_1',j_2'}:j_1'<j_2'\})$$ 
    and each $\theta_{j_1',j_2'}$ is defined as follows.
    It is the solution to the equation
    \begin{equation}\label{eqn:thetastar-m}
        \left(D(\theta_{j_1',j_2'})\right)^{n-1}\left(\theta_{j_1',j_2'}\cdot(\nu_{j_2'}-\nu_{j_1'})+t(\mu_{j_2'}-\mu_{j_1'})\right)=c_{j_2'}-c_{j_1'}
    \end{equation}
    if a solution exists, and it is $B+1$ otherwise.

    In addition, the equilibrium stated above is the unique symmetric natural-strategy Bayes Nash equilibrium, up to tie-breaking at threshold types and tie-breaking in \selection.
\end{theorem}
\begin{proof}
For every pair of actions $r<s$, write $\widehat \theta_{r,s}$ for the number defined by \eqref{eqn:thetastar-m}, with the convention $\widehat \theta_{r,s}=B+1$ if \eqref{eqn:thetastar-m} has no solution.
That is, if a solution exists, then
\[
    D(\widehat \theta_{r,s})^{n-1}\left(        \widehat \theta_{r,s}(\nu_s-\nu_r)        +t(\mu_s-\mu_r)\right)=c_s-c_r.
\]
The left-hand side is strictly increasing in $\theta$, so the solution is unique whenever it exists.

For notational simplicity, write $\tau_\ell:=\widehat\theta_{j_\ell,j_{\ell+1}}$ and  $\ell=1,\ldots,m'-1$, and set $\tau_0:=0$, and $\tau_{m'}:=B$.
Consider the symmetric strategy profile described in the theorem: every bidder uses the action strategy
\[
    a(\theta)=j_\ell\quad\text{for}\quad    \theta\in[\tau_{\ell-1},\tau_\ell),\qquad \ell=1,\ldots,m',
\]
with arbitrary tie-breaking at threshold types, and uses the natural bid $b(\theta)=\theta\nu_{a(\theta)}+t\mu_{a(\theta)}$.
We show that this symmetric strategy profile is a Bayes Nash equilibrium.

Fix a bidder $i$ and suppose all other bidders use the strategy above. 
Let $U_j(\theta)$ denote bidder $i$'s interim expected utility when her type is $\theta$, she chooses action $j$, and she submits the corresponding natural bid $\theta\nu_j+t\mu_j$.

We first show that two consecutive selected actions are tied exactly at their selected threshold.
Fix $\ell\in\{1,\ldots,m'-1\}$ and write
\[
    r:=j_\ell,\qquad s:=j_{\ell+1},\qquad \tau:=\tau_\ell=\widehat\theta_{r,s}.
\]
At type $\tau$, the natural bids from actions $r$ and $s$ are
\[
    \tau\nu_r+t\mu_r
    \qquad\text{and}\qquad
    \tau\nu_s+t\mu_s.
\]
Since $r$ and $s$ are consecutive selected actions, every opponent with type below $\tau$ plays some selected action weakly below $r$, and every opponent with type above $\tau$ plays some selected action weakly above $s$. 
Hence, for every opponent type $\theta'<\tau$, the opponent's bid is strictly below $\tau\nu_r+t\mu_r$, while for every opponent type $\theta'>\tau$, the opponent's bid is strictly above $\tau\nu_s+t\mu_s$.
Therefore, whether bidder $i$ chooses action $r$ or action $s$, she wins exactly when all other bidders' types are below $\tau$. 
This event has probability $D(\tau)^{n-1}$.
Moreover, conditional on this event, the expected second-price payment is the same under actions $r$ and $s$. 
Therefore,
\[
\begin{aligned}
    U_s(\tau)-U_r(\tau)
    =
    D(\tau)^{n-1}
    \left(
        \tau(\nu_s-\nu_r)+t(\mu_s-\mu_r)
    \right)
    -(c_s-c_r)=0,
\end{aligned}
\]
where the last equality follows from the definition of $\tau=\widehat\theta_{r,s}$.

By Lemma~\ref{prop:monotone_diff}, $U_s(\theta)-U_r(\theta)$ is increasing in $\theta$.
Thus,
\[
    U_r(\theta)>U_s(\theta)
    \quad\text{for}\quad \theta<\tau,
\]
and
\[
    U_s(\theta)>U_r(\theta)
    \quad\text{for}\quad \theta>\tau,
\]
up to possible equality at the threshold type itself.

Next, consider an action $h$ that is skipped by the selection algorithm when it jumps from $r$ to $s$, so that
\[
    r<h<s.
\]
Since $s$ is chosen by the selection algorithm as the first action that overtakes $r$, we have
\[
    \widehat\theta_{r,s}\le \widehat\theta_{r,h}.
\]
At $\tau=\widehat\theta_{r,s}$, the natural bid from action $h$ lies between the natural bids
from actions $r$ and $s$:
\[
    \tau\nu_r+t\mu_r
    <
    \tau\nu_h+t\mu_h
    <
    \tau\nu_s+t\mu_s.
\]
Hence, action $h$ has the same winning event and the same conditional expected payment as actions $r$ and $s$ at type $\tau$: bidder $i$ wins exactly when all opponents' types are below $\tau$. 
Consequently,
\[
\begin{aligned}
    U_h(\tau)-U_r(\tau)
    &=
    D(\tau)^{n-1}
    \left(
        \tau(\nu_h-\nu_r)+t(\mu_h-\mu_r)
    \right)
    -(c_h-c_r)  \\
    &\le 0,
\end{aligned}
\]
where the inequality follows from $\tau=\widehat\theta_{r,s}\le \widehat\theta_{r,h}$ and the monotonicity of
\[
    D(\theta)^{n-1}
    \left(
        \theta(\nu_h-\nu_r)+t(\mu_h-\mu_r)
    \right).
\]
Again by Lemma~\ref{prop:monotone_diff}, $U_h(\theta)-U_r(\theta)$ is increasing in $\theta$, so
\[
    U_h(\theta)\le U_r(\theta)
    \quad\text{for all}\quad \theta\le \tau.
\]
Similarly, since $h<s$ and $U_s(\tau)\ge U_h(\tau)$, Lemma~\ref{prop:monotone_diff}
implies
\[
    U_h(\theta)\le U_s(\theta)
    \quad\text{for all}\quad \theta\ge \tau.
\]
Thus, every action skipped between two consecutive selected actions is dominated by the lower selected action before the threshold and by the higher selected action after the threshold.

The same argument also covers the terminal case. If, at some step, the current action $r$ has no successor selected inside $[0,B]$, then for every $h>r$ we have $\widehat\theta_{r,h}>B$. 
At type $B$, choosing action $r$ or $h$ leads to the same winning event, namely winning with probability one, and the same expected second-price payment.
Therefore,
\[
    U_h(B)-U_r(B)
    =
    B(\nu_h-\nu_r)+t(\mu_h-\mu_r)-(c_h-c_r)
    <0.
\]
By Lemma~\ref{prop:monotone_diff}, $U_h(\theta)<U_r(\theta)$ for every $\theta<B$.
Hence, no higher action can be a profitable deviation after the last selected action.

Combining the preceding observations inductively over the selected thresholds
\[
    0=\tau_0<\tau_1<\cdots<\tau_{m'-1}<\tau_{m'}=B,
\]
we obtain that, for every $\ell=1,\ldots,m'$,
\[
    U_{j_\ell}(\theta)\ge U_h(\theta)
    \qquad
    \text{for every action }h\in[m]
    \text{ and every }
    \theta\in[\tau_{\ell-1},\tau_\ell).
\]
Therefore, the prescribed action $a(\theta)=j_\ell$ is a best response on $[\tau_{\ell-1},\tau_\ell)$. 
Given the chosen action, the natural bid is dominant in the second-price auction. Hence, the symmetric strategy profile constructed above is a symmetric Bayes Nash equilibrium.

It remains to prove uniqueness. 
Consider any symmetric natural-strategy Bayes Nash equilibrium.
By Theorem~\ref{thm:Bayesian-EquilibriaCharacterization-m}, its action rule must be a threshold strategy with actions ordered increasingly in type. 
Let two consecutive actions used in this equilibrium be $r<s$, and let $\tau$ be the threshold at which the equilibrium switches from $r$ to $s$.

At type $\tau$, the same argument as above applies: since no action strictly between $r$ and $s$ is used in a neighborhood of the threshold, choosing $r$ or $s$ yields the same winning event and the same conditional expected second-price payment. 
Therefore, equilibrium indifference at the threshold gives
\[
    D(\tau)^{n-1}
    \left(
        \tau(\nu_s-\nu_r)+t(\mu_s-\mu_r)
    \right)
    =c_s-c_r.
\]
Thus,
\[
    \tau=\widehat\theta_{r,s}.
\]

Moreover, if an action $h$ with $r<h<s$ is skipped between $r$ and $s$, then at the same threshold $\tau$ the action $h$ has the same winning event and same conditional expected payment as $r$ and $s$.
Since $h$ is not a strict best response at $\tau$, we must have
\[
    D(\tau)^{n-1}
    \left(
        \tau(\nu_h-\nu_r)+t(\mu_h-\mu_r)
    \right)
    \le c_h-c_r,
\]
which implies
\[
    \widehat\theta_{r,h}\ge \tau=\widehat\theta_{r,s}.
\]
Likewise, for any $h>s$, if choosing $h$ were strictly better than choosing $r$ at $\tau$, then $h$ would be a profitable deviation at the threshold. 
Therefore, $h$ is not strictly better than $r$ at $\tau$. 
Since choosing $h$ gives at least the payoff contribution
\[
    D(\tau)^{n-1}
    \left(
        \tau(\nu_h-\nu_r)+t(\mu_h-\mu_r)
    \right)
    -(c_h-c_r)
\]
relative to action $r$ on the event that all opponents' types are below $\tau$, it follows that
\[
    D(\tau)^{n-1}
    \left(
        \tau(\nu_h-\nu_r)+t(\mu_h-\mu_r)
    \right)
    \le c_h-c_r.
\]
Hence,
\[
    \widehat\theta_{r,h}\ge \tau=\widehat\theta_{r,s}.
\]
Thus, starting from the current action $r$, the next action $s$ used in any symmetric equilibrium is exactly the action selected by
\[
    \text{\selection}\left(\{\widehat\theta_{r',s'}:r'<s'\}\right),
\]
up to the deterministic tie-breaking rule of \selection.

Applying this argument recursively from the lowest action, the entire sequence of used actions and thresholds must coincide with the output of \selection. 
Therefore, the symmetric equilibrium constructed above is the unique symmetric natural-strategy Bayes Nash equilibrium, up to tie-breaking at threshold types and possible ties inside the selection algorithm.
\end{proof}

Finally, we move on and show the main result about the optimal choice $t\to\omega$ for the revenue maximization problem.

\begin{theorem}\label{thm:BNE-optimalt-m}
    Consider the second-price auction with a linear contract $T(q)=tq$.
    Let $R_n(t)$ be the expected revenue under the symmetric Bayes Nash equilibrium characterized in Theorem~\ref{thm:BNE-identical-m}. Then the limiting expected revenue
    $$\bar{R}(t):=\lim_{n\to\infty}R_n(t)$$
    is maximized at $t=\omega$.
\end{theorem}
\begin{proof}
Fix a reward factor $t\geq 0$.
For each $n\geq 2$, let $a_{n,t}(\theta)$ denote the action chosen by type $\theta$ in the symmetric equilibrium characterized in Theorem~\ref{thm:BNE-identical-m}, and write
\[
    H_n(\theta):=D(\theta)^{n-1}.
\]
Since the equilibrium bidding strategy is strictly increasing, a bidder with type $\theta$ wins with probability $H_n(\theta)$.
If she changes her action while keeping her equilibrium bid, her winning probability and expected auction payment remain unchanged.
Consequently, equilibrium optimality implies
\[
    a_{n,t}(\theta)
    \in
    \argmax_{j\in[m]}
    \left\{
        H_n(\theta)(\theta\nu_j+t\mu_j)-c_j
    \right\}.
\]
We include an additive constant $c_1$ in the argmax function for the ease of later calculations:
\[
    a_{n,t}(\theta)
    \in
    \argmax_{j\in[m]}
    \left\{
        H_n(\theta)(\theta\nu_j+t\mu_j)-(c_j-c_1)
    \right\}.
\]

We change variables from the bidder's type to her equilibrium winning probability.
For $\lambda\in(0,1]$, define
\[
    \theta_n(\lambda)
    :=
    D^{-1}\!\left(\lambda^{1/(n-1)}\right),
    \qquad
    J_{n,t}(\lambda)
    :=
    a_{n,t}\bigl(\theta_n(\lambda)\bigr).
\]
Then $H_n(\theta_n(\lambda))=\lambda$, so
\[
    J_{n,t}(\lambda)
    \in
    \argmax_{j\in[m]}
    \left\{
        \lambda\bigl(\theta_n(\lambda)\nu_j+t\mu_j\bigr)
        -(c_j-c_1)
    \right\}.
\]
For every fixed $\lambda>0$, we have $\theta_n(\lambda)\to B$.
Define the limiting action choice by
\[
    J_t(\lambda)
    \in
    \argmax_{j\in[m]}
    \left\{
        \lambda(B\nu_j+t\mu_j)-(c_j-c_1)
    \right\},
\]
with an arbitrary fixed tie-breaking rule.
The affine functions inside this maximum are pairwise distinct, since the action costs are strictly increasing.
Thus, ties occur at only finitely many values of $\lambda$.
At every other $\lambda$, the limiting maximizer is unique, and
\[
    J_{n,t}(\lambda)=J_t(\lambda)
    \qquad\text{for all sufficiently large }n.
\]

Next, we express the limiting revenue in terms of $J_t$.
Under the transformation $\theta\mapsto H_n(\theta)$, let
$f_{1,n}$ and $f_{2,n}$ denote the densities corresponding to
the highest and second-highest bidder types, respectively.
Since $D(\theta)$ is uniformly distributed on $[0,1]$,
the order-statistic formulas give
\[
    f_{1,n}(\lambda)
    =
    \frac{n}{n-1}\lambda^{1/(n-1)},
    \qquad
    f_{2,n}(\lambda)
    =
    n\left(1-\lambda^{1/(n-1)}\right).
\]
The auction payment is the second-highest natural bid, while
the winner contributes an additional expected payoff
$(\omega-t)\mu_{a_{n,t}(\theta)}$.
Therefore,
\[
    R_n(t)
    =
    \int_0^1
    \Bigl[
        f_{2,n}(\lambda)
        \bigl(
            \theta_n(\lambda)\nu_{J_{n,t}(\lambda)}
            +t\mu_{J_{n,t}(\lambda)}
        \bigr)
        +
        f_{1,n}(\lambda)(\omega-t)
        \mu_{J_{n,t}(\lambda)}
    \Bigr]\,d\lambda.
\]
For every $\lambda\in(0,1)$,
\[
    f_{1,n}(\lambda)\longrightarrow 1,
    \qquad
    f_{2,n}(\lambda)\longrightarrow-\log\lambda.
\]
Moreover, for all $n\geq2$,
\[
    0\leq f_{1,n}(\lambda)\leq2,
    \qquad
    0\leq f_{2,n}(\lambda)\leq2(-\log\lambda),
\]
where the second inequality follows from $1-e^{-x}\leq x$.
Since $t$ is fixed, the action set is finite, and
$\theta_n(\lambda)\in[0,B]$, the absolute value of the revenue integrand is bounded by a constant multiple of $1-\log\lambda$, which is integrable on $(0,1)$.
Dominated convergence therefore yields
\[
    \bar R(t)
    =
    \int_0^1
    \left[
        (-\log\lambda)
        \bigl(B\nu_{J_t(\lambda)}+t\mu_{J_t(\lambda)}\bigr)
        +(\omega-t)\mu_{J_t(\lambda)}
    \right]\,d\lambda.
\]

To simplify this expression, define
\[
    L_t(\lambda)
    :=
    \max_{j\in[m]}
    \left\{
        \lambda(B\nu_j+t\mu_j)-(c_j-c_1)
    \right\}.
\]
This function is continuous and piecewise linear, with
\[
    L_t'(\lambda)
    =
    B\nu_{J_t(\lambda)}+t\mu_{J_t(\lambda)}
    \qquad\text{for almost every }\lambda.
\]
Because $c_j>c_1$ for every $j>1$, action $1$ is the unique maximizer for all sufficiently small positive $\lambda$.
Hence,
\[
    L_t(\lambda)=\lambda(B\nu_1+t\mu_1)
    \qquad\text{near }\lambda=0,
\]
and in particular
$\lim_{\lambda\to0^+}(-\log\lambda)L_t(\lambda)=0$.
Integration by parts gives
\[
\begin{aligned}
    \int_0^1(-\log\lambda)L_t'(\lambda)\,d\lambda
    &=
    \lim_{\varepsilon\to0^+}
    \left(
        \left[(-\log\lambda)L_t(\lambda)\right]_\varepsilon^1
        +\int_\varepsilon^1\frac{L_t(\lambda)}{\lambda}\,d\lambda
    \right)\\
    &=
    \int_0^1\frac{L_t(\lambda)}{\lambda}\,d\lambda.
\end{aligned}
\]
Substituting into the limiting-revenue expression, we obtain
\[
\begin{aligned}
    \bar R(t)
    &=
    \int_0^1
    \left[
        \frac{L_t(\lambda)}{\lambda}
        +(\omega-t)\mu_{J_t(\lambda)}
    \right]\,d\lambda\\
    &=
    \int_0^1
    \left[
        B\nu_{J_t(\lambda)}
        +\omega\mu_{J_t(\lambda)}
        -\frac{c_{J_t(\lambda)}-c_1}{\lambda}
    \right]\,d\lambda.
\end{aligned}
\]

Finally, when $t=\omega$, the definition of $J_\omega(\lambda)$
implies that it maximizes the last integrand pointwise:
\[
    J_\omega(\lambda)
    \in
    \argmax_{j\in[m]}
    \left\{
        B\nu_j+\omega\mu_j-\frac{c_j-c_1}{\lambda}
    \right\}.
\]
Consequently, for every $t\geq0$,
\[
\begin{aligned}
    \bar R(t)
    &\leq
    \int_0^1
    \max_{j\in[m]}
    \left\{
        B\nu_j+\omega\mu_j-\frac{c_j-c_1}{\lambda}
    \right\}\,d\lambda\\
    &=\bar R(\omega).
\end{aligned}
\]
Thus, the limiting expected revenue is maximized at $t=\omega$.
\end{proof}

\section{Proof of Theorem~\ref{thm:revenue-equivalence} for General $m$}
\label{sec:revenue-equivalence-m}
We first show the following proposition that extends Proposition~\ref{prop:strictlyIncreasingBidding->threshold}.

\begin{proposition}\label{prop:strictlyIncreasingBidding->threshold-m}
    Consider identical bidders' type distributions and any linear contract $T(q)=tq$.
    Suppose there exists a symmetric Bayes Nash equilibrium where everyone plays $(a,b)$ with strictly increasing $b$.
    Then the action strategy $a$ must be a threshold strategy with available action sets $A=\{j_1,\ldots,j_{m'}\}$ (where $j_1<j_2<\cdots<j_{m'}$) and thresholds $\theta_{j_1,j_2},\theta_{j_2,j_3},\ldots,\theta_{j_{m'}-1,j_{m'}}$, where $A$ and the thresholds are given by \selection$(\{\theta_{j_1',j_2'}:j_1'<j_2'\})$ and each $\theta_{j_1',j_2'}$ is defined as follows.
    It is the solution to the equation
    \begin{equation}\label{eqn:thetastar-increasingbidding-m}
        (D(\theta_{j_1',j_2'}))^{n-1}\left(\theta_{j_1',j_2'}\cdot(\nu_{j_2'}-\nu_{j_1'})+t(\mu_{j_2'}-\mu_{j_1'})\right)=c_{j_2'}-c_{j_1'}
    \end{equation}
    if a solution exists, and it is $B+1$ otherwise.
\end{proposition}
\begin{proof}
    Fix everyone's bidding strategy to be $b$.
    Consider bidder $i$'s type $\theta$ and any two actions $j_1'$ and $j_2'$ with $j_1'<j_2'$.
    By the same analysis in the proof of Proposition~\ref{prop:strictlyIncreasingBidding->threshold}, playing $j_1'$ is strictly better than $j_2'$ if $\theta<\theta_{j_1',j_2'}$, and playing $j_2'$ is strictly better if $\theta>\theta_{j_1',j_2'}$, where $\theta_{j_1',j_2'}$ is defined by Eqn.~(\ref{eqn:thetastar-increasingbidding-m}).
    This is the property sufficient for proving Theorem~\ref{thm:BNE-identical-m},\footnote{In fact, this property is stronger than what we have when proving Theorem~\ref{thm:BNE-identical-m}: in the context of Theorem~\ref{thm:BNE-identical-m}, the threshold characterized by Eqn.~(\ref{eqn:thetastar-increasingbidding-m}) does not hold if $j_1'$ and $j_2'$ are not in the available action set $A$, because $(D(\theta_{j_1',j_2'}))^{n-1}$ is not the probability of winning when bidder $i$'s type is $\theta_{j_1',j_2'}$ and is playing $j_1'$ or $j_2'$ (notice that the other agents are playing a third action other than $j_1'$ and $j_2'$). On the other hand, in the context of this proposition, everyone is adopting the same bidding strategy. The winning probability of bidder $i$ is always $D(\theta_{j_1',j_2'})^{n-1}$.} and the remaining part of the proof is the same as the proof of Theorem~\ref{thm:BNE-identical-m}.
\end{proof}

To prove Theorem~\ref{thm:revenue-equivalence}, we again show the same two observations:
\begin{enumerate}
        \item the action strategy $a$ must be the same as it is in $\Sigma_{\spa}$, and
        \item for each bidder $i$ and each $\theta^{(i)}\in\Theta$, the expected payment (where the expectation is taken over $\theta^{(-i)}\sim\mathcal{D}^{(-i)}$) of bidder $i$ under $\Sigma$ in $\mathcal{M}$ is the same as the expected payment under $\Sigma_{\spa}$ in $\mathcal{M}_{\spa}$.
\end{enumerate}

The first claim follows from Proposition~\ref{prop:strictlyIncreasingBidding->threshold-m} and Theorem~\ref{thm:BNE-identical-m}.
The second claim can be proved in the same way as it is in the $m=2$ case: instead of analyzing $P(\theta)$ in two regimes, $[0,\theta^\ast]$ and $[\theta^\ast,B]$, we make the same analysis of $P(\theta)$ in $m'$ regimes separated by the thresholds $\theta_{j_1,j_2},\theta_{j_2,j_3},\ldots,\theta_{j_{m'}-1,j_{m'}}$ as described in Proposition~\ref{prop:strictlyIncreasingBidding->threshold-m}.

\section{First-Price Auction: Extensions to $m$ Actions}\label{append:fpa}
\begin{theorem}[Existence of Symmetric BNE for Identical Type Distributions under FPA]
    Consider the first-price auction when bidders' type distributions are identical (following Assumption~\ref{assumption:distribution}). 
    Suppose a linear contract $T(q)=tq$ is used and there are $m$ actions.
    Then there exists a symmetric Bayes Nash equilibrium $\Sigma=(\sigma^{(1)},\ldots,\sigma^{(n)})$ where every bidder $i\in[n]$ adopts the same strategy $\sigma^{(i)}=(a,b)$.
    The action strategy $a$ follows the one in Proposition~\ref{prop:strictlyIncreasingBidding->threshold-m}.

    Each bidder adopts a continuous and strictly increasing bidding strategy $b(\theta)$, which is given by
    $$
    b(\theta)
    =
    \theta\nu_{a(\theta)}+t\mu_{a(\theta)}
    -
    \frac{
        c_{a(\theta)}-c_1
        +
        \displaystyle\int_0^\theta
        D(z)^{n-1}\nu_{a(z)}\,dz
    }{
        D(\theta)^{n-1}
    }
    $$
    for every $\theta>0$, and set $b(0):=t\mu_1$ for $\theta=0$.
\end{theorem}
\begin{proof}
    Proposition~\ref{prop:strictlyIncreasingBidding->threshold-m} gives the only possible characterization of action strategy under any equilibrium.
    Again, it remains to verify the three claims in the proof of Theorem~\ref{thm:BNE-identical-fpa}.
    Claim 1 is directly implied by Proposition~\ref{prop:strictlyIncreasingBidding->threshold-m}, where the bidding strategy can be obtained according to Footnote~\ref{fn:fpa-b}.
    Claim 2 holds due to the same reason as $m=2$, as the bidding strategy is continuous and strictly increasing on each interval with a constant action, and continuous at each threshold $\theta_{j_1,j_2}$ given in Eqn.~(\ref{eqn:thetastar-increasingbidding-m}).
    Lemma~\ref{lem:fpa-V} can be generalized to $m$ actions, that $$\frac{\partial \mathcal V(\theta,x)}{\partial\theta}=x\nu_{a(\theta,x)}$$ holds almost everywhere except those types where the derivatives may not exist, and is weakly increasing in $x$ for $m$ actions.
    The reason is the same as its original proof, where increasing $x$ only makes an action with a higher cost more attractive.
    Therefore, claim 3 holds, as the subsequent argument is the same for $m$ actions as for two actions.
\end{proof}

\end{document}